\documentclass[11pt,a4paper]{article}

\usepackage[utf8]{inputenc}
\usepackage[T1]{fontenc}
\usepackage{amsmath,amssymb,amsthm}
\usepackage{mathtools}
\usepackage{booktabs}
\usepackage{float}
\usepackage{array}
\usepackage{tabularx}
\usepackage{multirow}
\usepackage{graphicx}
\usepackage{xcolor}
\usepackage{enumitem}
\usepackage[hidelinks]{hyperref}
\usepackage{algorithm}
\usepackage{algpseudocode}
\usepackage{natbib}
\usepackage[margin=1in]{geometry}
\newcolumntype{Y}{>{\raggedright\arraybackslash}X}
\newcolumntype{P}[1]{>{\raggedright\arraybackslash}p{#1}}
\usepackage{caption}
\usepackage{subcaption}
\usepackage{tikz}
\usetikzlibrary{positioning,fit,arrows.meta,calc}

\newtheorem{theorem}{Theorem}
\newtheorem{corollary}[theorem]{Corollary}
\newtheorem{definition}{Definition}
\newtheorem{proposition}{Proposition}

\newcommand{\CES}{\text{CES}}

\newcommand{\R}{\mathbb{R}}
\newcommand{\E}{\mathbb{E}}
\newcommand{\Prob}{\mathbb{P}}
\DeclareMathOperator*{\argmin}{arg\,min}
\DeclareMathOperator*{\argmax}{arg\,max}

\title{%
  Multi-Method Causal Evidence Synthesis:\\
  Ranking Candidate Drivers by Convergent\\
  Cross-Method Evidence from Observational Data%
}

\author{
  Manish Gupta \\
  Tricon Infotech \\
  \texttt{manishg@triconinfotech.com}
  \and
  Dipanjan De \\
  Tricon Infotech \\
  \texttt{dipanjan.de@triconinfotech.com}
}

\date{August 2026}

\providecommand{\ensemblePfive}{$1.0$}
\providecommand{\ensemblePten}{$0.96$}
\providecommand{\ensembleFOne}{$0.686 \pm 0.035$}
\providecommand{\ensembleVsBest}{$-0.029$}
\providecommand{\bestIndivFOne}{$0.714$}
\providecommand{\maxLooDrop}{$0.032$}
\providecommand{\weightSpearman}{$0.9449$}
\providecommand{\weightSpearmanMin}{$0.814$}
\providecommand{\decompPthreeWithout}{$0.3333$}
\providecommand{\decompPthreeWith}{$1.0$}
\providecommand{\nlEnsembleRecall}{$0.667$}
\providecommand{\calBaseRate}{$0.032$}
\providecommand{\calECE}{$0.101 \pm 0.002$}
\providecommand{\calECEcal}{$0.004 \pm 0.000$}
\providecommand{\calBrierRaw}{$0.025 \pm 0.001$}
\providecommand{\calBrierCal}{$0.011 \pm 0.001$}
\providecommand{\calFolds}{5}
\providecommand{\meanMethodCorr}{$0.209$}
\providecommand{\maxMethodCorr}{$0.914$}
\providecommand{\grangerTeCorr}{$0.075$}
\providecommand{\fpRaw}{$0.277 \pm 0.338$}
\providecommand{\fpBH}{$0.239 \pm 0.335$}
\providecommand{\fpEnsemble}{$0.000 \pm 0.001$}
\providecommand{\fpEnsembleThree}{$0.0051 \pm 0.0033$}
\providecommand{\tpEnsembleFour}{$0.50 \pm 0.05$}
\providecommand{\tpEnsembleThree}{$0.69 \pm 0.09$}
\providecommand{\eThreeShortTF}{$0.60 \pm 0.07$}
\providecommand{\eThreeLongTF}{$0.69 \pm 0.03$}
\providecommand{\eThreeNarrowNF}{$0.33 \pm 0.12$}
\providecommand{\eThreeWideNF}{$0.71 \pm 0.00$}

\providecommand{\hcF}{$1.000$}
\providecommand{\hcBest}{$1.000$}
\providecommand{\hcEdges}{10}
\providecommand{\mfgF}{$0.737$}
\providecommand{\mfgBest}{$0.842$}
\providecommand{\mfgEdges}{9}
\providecommand{\sachsPfive}{$1.0$}
\providecommand{\liftHigh}{$0.100 \pm 0.096$}
\providecommand{\liftNull}{$-0.000 \pm 0.009$}
\providecommand{\liftSpearman}{$0.119$}

\providecommand{\cesPrimaryFBl}{$0.69 \pm 0.04$}
\providecommand{\ntNonlinF}{$0.42$}
\providecommand{\cesNonlinF}{$0.74$}

\providecommand{\stackedF}{$0.71 \pm 0.00$}
\providecommand{\stackedCesF}{$0.69 \pm 0.04$}
\providecommand{\nullModRate}{$0.0047 \pm 0.0022$}
\providecommand{\nullMaxCes}{$0.40 \pm 0.09$}
\providecommand{\renormSpearman}{$1.000 \pm 0.000$}
\providecommand{\renormFOne}{$0.686 \pm 0.035$}
\providecommand{\zeroFillFOne}{$0.686 \pm 0.035$}

\providecommand{\bnNetCount}{6}

\providecommand{\ensemblePfive}{TBD}\providecommand{\ensemblePten}{TBD}
\providecommand{\ensembleFOne}{TBD}\providecommand{\bestIndivFOne}{TBD}
\providecommand{\ensembleVsBest}{TBD}\providecommand{\maxLooDrop}{TBD}
\providecommand{\weightSpearman}{TBD}\providecommand{\weightSpearmanMin}{TBD}
\providecommand{\decompPthreeWithout}{TBD}\providecommand{\decompPthreeWith}{TBD}
\providecommand{\nlEnsembleRecall}{TBD}\providecommand{\calBaseRate}{TBD}
\providecommand{\calECE}{TBD}\providecommand{\meanMethodCorr}{TBD}
\providecommand{\calECEcal}{TBD}\providecommand{\calBrierRaw}{TBD}\providecommand{\calBrierCal}{TBD}
\providecommand{\maxMethodCorr}{TBD}\providecommand{\grangerTeCorr}{TBD}
\providecommand{\fpRaw}{TBD}\providecommand{\fpBH}{TBD}\providecommand{\fpEnsemble}{TBD}
\providecommand{\hcF}{TBD}\providecommand{\hcBest}{TBD}\providecommand{\hcEdges}{TBD}
\providecommand{\mfgF}{TBD}\providecommand{\mfgBest}{TBD}\providecommand{\mfgEdges}{TBD}
\providecommand{\sachsPfive}{TBD}
\providecommand{\liftHigh}{TBD}\providecommand{\liftNull}{TBD}
\providecommand{\liftSpearman}{TBD}

\providecommand{\bnNetCount}{TBD}

\begin{document}

\maketitle

\begin{abstract}
Causal inference from observational data is ubiquitous across science and industry, yet practitioners overwhelmingly rely on a single analytical method, typically regression or SHAP-based feature importance, and treat its output as causal truth. Recent automated tools \emph{select} an optimal method for a given dataset \citep{causalcopilot2025, causaltune2024, opportunityfinder2023}, but comparatively little work \emph{synthesizes} evidence from methods spanning different mathematical traditions applied to the same raw data. We present Multi-Method Causal Evidence Synthesis (MCES), a framework whose goal is to rank \emph{which candidate drivers} in an observational system are most likely relevant to a set of outcome variables, and with what strength of evidence. MCES runs eleven analytical methods across eight distinct mathematical traditions on observational panel data and pools their outputs into a Convergent Evidence Score (CES) that ranks driver--outcome relationships. CES quantifies the \emph{convergence of evidence} across analytical lenses: the degree to which methods with different assumptions point toward the same driver--outcome relationship. It does not claim causal identification in the interventionist sense; it supports hypothesis prioritization, not a universally transferable probability of causation. We additionally fit a \emph{scenario-specific} empirical calibration of CES against known ground truth and are explicit that its transportability to new domains is not established. The framework first applies \emph{Structural--Behavioral Decomposition} to remove definitional (algebraic) relationships from the analysis space (but only where identity components are themselves candidate drivers, the setting in which such false positives actually arise), then runs all methods, normalizes outputs to $[0,1]$, and computes a weighted composite in the form of a linear opinion pool \citep{stone1961, genest1986}. We distinguish MCES from method selection, same-method ensembles, cross-algorithm prediction ensembles (Super Learner), and literature-level synthesis. Using synthetic data with embedded ground truth, the real Sachs protein-signaling benchmark, a suite of standard Bayesian-network structure benchmarks, and two additional synthetic domains (healthcare, manufacturing), we show that MCES reliably ranks the true edges near the top: on the primary synthetic scenario the highest-scoring pairs are all true edges (Precision@5 $=$ \ensemblePfive, Precision@10 $=$ \ensemblePten), while showing a low empirical rate of null pairs reaching Moderate-or-higher convergence on the evaluated scenarios; per-method significance gates are Benjamini--Hochberg adjusted across the driver--outcome grid by default. Our central methodological point is not that the pool beats any individual method on a fixed accuracy metric (indeed, on some scenarios a single well-chosen method scores higher) but that \emph{no single method is uniformly best across the evaluated scenarios}, so MCES offers a method-agnostic default that avoids committing in advance to a single analytical tradition. We report this even-handedly.
\end{abstract}


\section{Introduction}\label{sec:intro}

\subsection{The Ubiquity of the Problem}

Organizations across industries routinely face a fundamental question: \emph{``Why did this metric change, and what caused it?''} Whether a manufacturer detects rising defect rates after changing suppliers, a school district sees test scores fall after a curriculum change, or a SaaS company sees falling retention after a product update, the diagnostic challenge is the same. Observational data is abundant; causal understanding is scarce.

Table~\ref{tab:domains} illustrates the universality of this problem across domains.

\begin{table}[ht]
\centering
\caption{The causal diagnostic problem across domains.}
\label{tab:domains}
\small
\begin{tabularx}{\textwidth}{@{}lYYll@{}}
\toprule
\textbf{Domain} & \textbf{Drivers ($X$)} & \textbf{Outcomes ($Y$)} & \textbf{Units} & \textbf{Time} \\
\midrule
Enterprise Software & Product perf., training, usability, staffing & Revenue, throughput, efficiency & Sites & Weeks \\
Healthcare & Treatments, protocols, staffing, equipment & Mortality, readmission, satisfaction & Hospitals & Months \\
Manufacturing & Machine settings, materials, operator skill & Defect rate, yield, throughput & Lines & Shifts \\
Education & Teaching methods, class size, technology & Test scores, graduation rate & Schools & Semesters \\
SaaS / Product & Features, UX changes, pricing, onboarding & Retention, NPS, revenue, churn & Cohorts & Months \\
\bottomrule
\end{tabularx}
\end{table}

Decisions are frequently made on weak correlational evidence, and the consequences of misattribution (acting on the wrong candidate relationship, or missing the true one) are substantial and often invisible.

\subsection{Common Analytical Practice}\label{sec:tiers}

Analysts approach the driver-outcome question with three broad families of method, each answering a different question.

\paragraph{Descriptive and associational analysis.}
Correlation, linear regression, and standard reporting. These identify \emph{where} an outcome moved and which variables co-move with it, but not \emph{why}: association is not causation, and the direction of any relationship is left open.

\paragraph{Predictive modeling.}
Gradient-boosted models with feature-importance or SHAP attributions \citep{lundberg2017}. These capture non-linear structure and rank features by contribution to a prediction, but importance for a \emph{prediction} is not evidence of \emph{causation} of the outcome, and the two are routinely conflated in practice.

\paragraph{Causal inference and effect estimation.}
Methods that target an identified causal estimand (difference-in-differences, instrumental variables, structural causal models, and tools such as DoWhy and EconML). These can support causal claims under stated assumptions, but are typically applied one method at a time, often require a manually specified causal graph or design, and produce output that requires expert interpretation.

Each family has genuine strengths and characteristic blind spots. In practice an analysis usually commits to one method, and thereby to that method's assumptions and failure modes, even though the appropriate choice is rarely known in advance.

\subsection{The Single-Method Trap}\label{sec:trap}

Current practice selects one analytical method and treats its output as ground truth. This fails because each method has different assumptions, captures different types of relationships, and has different failure modes. Table~\ref{tab:blindspots} illustrates these complementary blind spots.

\begin{table}[ht]
\centering
\caption{Complementary blind spots of component analytical methods.}
\label{tab:blindspots}
\small
\begin{tabular}{@{}lll@{}}
\toprule
\textbf{Method} & \textbf{Captures} & \textbf{Misses} \\
\midrule
Regression & Linear associations & Non-linear effects, causal direction \\
SHAP & Non-linear importance & Causal direction, temporal dynamics \\
Granger Causality & Temporal predictability & Non-linear relationships \\
Bayesian Networks & DAG structure & Hidden confounders \\
Causal Forest & Heterogeneous effects & Median binarization (impl.\ choice) \\
\bottomrule
\end{tabular}
\end{table}

A concrete example illustrates the danger: when average price and volume are included among the candidate features, SHAP can rank the pair (Average~Price, Revenue) highly. This is trivially true because $\text{Revenue} = \text{Price} \times \text{Volume}$ by definition, it is an algebraic identity, not a causal discovery. The pathology arises specifically when an outcome's algebraic \emph{components} are also present in the feature set; Section~\ref{sec:decomp} formalizes this condition, and Section~\ref{sec:exp_decomp} shows that removing such identity pairs raises top-of-list precision from \decompPthreeWithout\ to \decompPthreeWith\ in a controlled setting.

\subsection{From Triangulation to Quantification}\label{sec:triangulation}

The concept of using multiple methods to strengthen causal inference, \emph{causal triangulation}, has been recommended in the methodological literature for over a decade \citep{lawlor2016, munafio2018, hammerton2021}. Related lines of work vary an analysis over many defensible choices: specification-curve analysis \citep{simonsohn2020} and multiverse analysis \citep{steegen2016} report the distribution of one estimand across specifications, and evidence factors \citep{rosenbaum2017} combine quasi-independent tests of a single hypothesis. These differ from our setting in two ways: they typically concern one estimand or one hypothesis, and they are not organized to \emph{rank many candidate relationships} by cross-method agreement. In the triangulation literature specifically, the recommendation has remained largely \textbf{qualitative} (``use multiple methods and see if they agree''), without a concrete rule for which heterogeneous evidence measures to combine, how to normalize outputs that target different quantities, or how to form a composite ranking.

Recent work has advanced along adjacent but distinct lines:
\begin{itemize}[nosep]
  \item \citet{bhattacharya2026} formalize \emph{robust weighted triangulation of causal effects under model uncertainty}, combining identified effect functionals from multiple candidate causal models, each with its own identifying assumptions, weighted by data-driven measures of model validity, with error bounds and valid inference. This is the closest formal treatment of weighted triangulation, and the contrast sharpens our scope: they combine estimates of a \emph{single common estimand} (the causal effect) across competing models, whereas MCES pools \emph{non-commensurable} evidence measures that target different quantities and uses them to \emph{rank many} candidate driver--outcome relationships rather than to estimate one effect.
  \item \citet{shi2025} quantified triangulation at the \emph{literature} level, extracting evidence from published papers using LLMs.
  \item Causal-Copilot \citep{causalcopilot2025} and CausalTune \citep{causaltune2024} \emph{select} the best method for a dataset.
  \item A 2025--2026 line of work ensembles \emph{multiple causal-discovery algorithms} into a single structural estimate: voting-theoretic aggregation with recovery conditions \citep{vo2026voting}, linear-opinion-pooled discovery experts with LLM reweighting \citep{li2026cea}, and structure-learning ensembles for prioritized health hypotheses \citep{adhikari2025heterogeneous}. These pool within the causal-discovery family to recover a graph; Section~\ref{sec:related} details how our setting differs in both members and task.
\end{itemize}

These differ from the synthesis we pursue here: \emph{selection discards evidence from all non-selected methods; synthesis aggregates evidence from all of them.} We are careful to distinguish this from cross-algorithm \emph{prediction} ensembles such as Super Learner \citep{vanderlaan2007}, which optimize a single predictive loss for one estimand via cross-validated weights; MCES instead pools evidence across methods that target \emph{different} estimands (association, temporal predictability, graphical dependence, treatment effect), with weights that encode epistemic priors rather than minimizing a shared loss (Section~\ref{sec:related}).

\subsection{Contributions}

We make the following contributions:

\begin{enumerate}[nosep]
  \item \textbf{MCES Framework}, a quantitative framework that combines heterogeneous evidence measures which do \emph{not} estimate a common causal estimand (association, temporal predictability, graphical dependence, treatment effect), for the purpose of ranking many candidate driver--outcome relationships by cross-method agreement. The pooling rule itself is classical, and ensembling \emph{within} the causal-discovery family is now established \citep{vo2026voting, li2026cea}; the combination we claim as novel is narrower and specific: non-commensurable evidence from causal \emph{and non-causal} traditions, pooled to rank a declared driver--outcome grid rather than to recover a graph. We distinguish this from method selection, same-method and cross-algorithm structural ensembles, and specification/multiverse analyses in Section~\ref{sec:related}.
  \item \textbf{Structural-Behavioral Decomposition}, formal separation of definitional (algebraic) from behavioral (causal) relationships as a pre-processing step for causal analysis.
  \item \textbf{Convergent Evidence Score (CES)}, a weighted composite score in $[0,1]$ that quantifies convergence of evidence across analytical traditions. We name the score for what it measures, convergence of evidence bearing on a causal hypothesis, not for a causal guarantee it does not provide. MCES uses uniform weights by default; optional tiered weights encode an application-specific evidential preference and are studied by sensitivity analysis.
  \item \textbf{Theoretical analysis}, a variance-based characterization of how lower cross-method score correlation can improve the stability of the pooled score under stated assumptions.
  \item \textbf{Empirical evaluation}, experiments on synthetic ground truth, the real Sachs benchmark, and two additional synthetic domains, characterizing \emph{when} the ensemble helps (robustness to scenario shift, non-linear detection) and when it does not (it is not uniformly the single best predictor), together with calibration and false-positive analyses.
  \item \textbf{Calibration and multiple-testing controls}, scenario-specific isotonic calibration of CES, and default Benjamini--Hochberg adjustment of the applicable per-method significance tests across the driver--outcome grid.
  \item \textbf{Reference implementation}, a Python package with three synthetic domains and a real-data benchmark loader (available from the authors on request).
\end{enumerate}

\section{Related Work}\label{sec:related}

\subsection{Causal Triangulation}

Methodological triangulation originates with \citet{denzin1970} in social science. \citet{lawlor2016} formalized it for aetiological epidemiology, proposing that multiple analytical approaches with different assumptions should be applied to the same research question. \citet{munafio2018} elevated this to a general principle in \emph{Nature}: ``Robust research needs many lines of evidence.'' \citet{hammerton2021} extended the framework with practical guidance in \emph{Psychological Medicine}.

All of these remain \textbf{qualitative}: they recommend using multiple methods but do not specify how to combine their outputs into a quantitative measure. MCES operationalizes this decade-old recommendation.

\subsection{Component Analytical Methods}

Each of the eleven methods in MCES has a rich individual literature. Table~\ref{tab:methods_lit} summarizes the foundational work and domain applications.

\begin{table}[ht]
\centering
\caption{Foundational literature for MCES component methods.}
\label{tab:methods_lit}
\small
\begin{tabular}{@{}lP{3.3cm}P{5.2cm}@{}}
\toprule
\textbf{Method} & \textbf{Key Papers} & \textbf{Limitation} \\
\midrule
Partial Correlation & \citet{pearson1895} & Cannot determine direction \\
Lasso Regression & \citet{tibshirani1996} & Assumes linearity \\
Distance Correlation & \citet{szekely2007} & Unsigned; no direction \\
Predictive Power Score & \citet{wetschoreck2020} & Univariate; no confounder control \\
Mixed-Effects Models & \citet{laird1982} & Assumes linearity \\
Random Forest + SHAP & \citet{breiman2001, lundberg2017} & Importance $\neq$ causation \\
Granger Causality & \citet{granger1969} & Assumes linearity, stationarity \\
Interrupted Time Series & \citet{box1975, bernal2017} & Requires known intervention \\
Transfer Entropy & \citet{schreiber2000} & Requires long time series \\
Bayesian Networks & \citet{pearl2000, spirtes2000} & Assumes causal sufficiency \\
Causal Forests & \citet{wageracthey2018} & Median binarization (impl.\ choice) \\
\bottomrule
\end{tabular}
\end{table}

Each method has been individually applied to problems in many fields. To our knowledge, no published work pools methods from all of these traditions into a single evidence score computed on raw data for the purpose of ranking candidate relationships.

\subsection{Ensemble Causal Discovery (Same-Method)}

Existing ensemble approaches largely aggregate instances of the \emph{same} algorithm across data partitions:
\begin{itemize}[nosep]
  \item The FGES/TETRAD line of work \citep{ramsey2018} bootstraps a causal-discovery algorithm and votes on edge presence.
  \item \citet{guo2021} propose scalable two-phase causality ensembles combining an algorithm across data partitions.
  \item E-CIT \citep{ecit2025} partitions data, runs a base conditional independence test on each subset, and aggregates the resulting statistics.
  \item Stability selection \citep{meinshausen2010} subsamples the data and retains variables selected with high frequency by a single base selector.
\end{itemize}

\textbf{Distinction:} these combine instances of the same algorithm (or selector) class across different data splits, so they share that algorithm's structural assumptions and failure modes.

\paragraph{Cross-algorithm structural ensembles (2025--2026).} A recent and fast-moving line of work aggregates \emph{multiple different causal-discovery algorithms} into one structural estimate, and it must be distinguished carefully from what we do. \citet{vo2026voting} give a voting-theoretic framework for ensembling structural predictions, with conditions under which the aggregate recovers the true graph. \citet{li2026cea} pool causal-discovery experts via linear opinion pooling, the same aggregation rule we use, with an LLM reweighting experts near decision boundaries; \citet{peng2026caution} similarly resolve most edges by algorithmic consensus and arbitrate the rest with a trust-calibrated LLM. \citet{adhikari2025heterogeneous} combine an ensemble of structure-learning algorithms with heterogeneous effect estimation to produce robust prioritized causal hypotheses in healthcare. Quantitative ensembling \emph{within} the causal-discovery family is therefore established, including with theoretical guarantees, and we claim no novelty for the pooling rule itself.

The distinction that remains, and that defines this paper's contribution, is twofold. First, the \emph{members}: every expert in those ensembles is a causal-discovery algorithm estimating the same object, a graph, whereas MCES deliberately pools evidence measures that estimate \emph{different} objects, association, predictive importance, temporal predictability, graphical dependence, and treatment contrast, including methods that are explicitly not causal. Second, the \emph{task}: those works recover graph structure, whereas MCES ranks a predefined driver--outcome candidate grid for hypothesis prioritization, with empirical false-positive control in place of recovery guarantees. To our knowledge, this specific combination, non-commensurable cross-tradition evidence pooled to rank a declared candidate grid, has not been studied, though we do not claim that no applied study has ever informally compared methods.

\subsection{Automated Causal Inference (Method Selection)}

Recent tools automate causal analysis by \emph{selecting} the optimal method:
\begin{itemize}[nosep]
  \item \textbf{Causal-Copilot} \citep{causalcopilot2025}: an LLM-powered autonomous agent integrating 20+ methods. It analyzes data characteristics to select the appropriate algorithm, configure hyperparameters, and interpret results. It does \emph{not} synthesize across methods.
  \item \textbf{CausalTune} \citep{causaltune2024}: AutoML for causal estimators using energy scoring to select the best estimator. It does not combine all estimators.
  \item \textbf{OpportunityFinder} \citep{opportunityfinder2023}: a code-less framework for panel data that dynamically selects algorithms based on data scale.
\end{itemize}

The critical distinction is between \emph{selection} and \emph{synthesis}. Selection picks the single best method for the data and discards evidence from all other methods. Synthesis runs all methods and combines their evidence. Selection preserves the blind spots of the chosen method; synthesis mitigates them through complementarity.

\subsection{Evidence Synthesis from Literature}

\citet{shi2025} present the ``Evidence Triangulator'' in \emph{Nature Communications}, which uses LLMs to extract causal evidence from published papers and computes Convergency of Evidence (CoE) and Level of Convergency (LoC) scores. This is the closest existing work to our concept. However, it operates at the \emph{meta-analysis} level, it requires published studies as input, not raw data. Our CES operates on raw observational data: we run the methods ourselves rather than extracting results from literature.

\subsection{Outcome Decomposition}

DuPont Analysis (1920s), KPI trees, and multiplicative decomposition are standard tools for algebraic outcome breakdown, and practitioner treatments of ``why did the KPI change'' decomposition are common. These are used for attribution \emph{within} an identity; they are not combined with statistical causal discovery to \emph{screen out} identity edges before analysis. Our Structural--Behavioral Decomposition serves that specific screening role.

\subsection{Summary: What Exists vs.\ Our Gap}

Table~\ref{tab:gap} summarizes the landscape.

\begin{table}[ht]
\centering
\caption{The 4-way distinction: existing paradigms vs.\ MCES.}
\label{tab:gap}
\small
\begin{tabularx}{\textwidth}{@{}llYY@{}}
\toprule
\textbf{Paradigm} & \textbf{Examples} & \textbf{What It Does} & \textbf{Limitation} \\
\midrule
Method Selection & CausalTune, Causal-Copilot & Picks best ONE method & Discards all other evidence \\
Same-Method Ensemble & TETRAD, E-CIT & Same algorithm on data splits & Shared assumptions \& failure modes \\
Structural Ensemble & \citet{vo2026voting}, \citet{li2026cea} & Pools causal-discovery algorithms into one graph & Members share the graph-recovery task and causal-family assumptions \\
Literature Synthesis & Shi et al.\ (2025) & LLM extracts from papers & Requires existing studies \\
\textbf{Cross-Tradition (MCES)} & \textbf{This paper} & \textbf{11 methods, 8 traditions, incl.\ non-causal; ranks a declared grid} & \textbf{No recovery guarantee; empirical FP control only} \\
\bottomrule
\end{tabularx}
\end{table}

\section{Problem Formulation}\label{sec:problem}

\subsection{General Setup}

Consider $N$ observational units (hospitals, factories, schools) indexed by $i = 1, \ldots, N$, observed over $T$ time periods indexed by $t = 1, \ldots, T$. Let $\mathbf{X} = \{x_1, \ldots, x_M\}$ denote $M$ potential driver variables, and $\mathbf{Y} = \{y_1, \ldots, y_K\}$ denote $K$ outcome variables. At each observation point $(i,t)$, we observe the vector:
\begin{equation}
  \mathbf{O}_{i,t} = \big(\mathbf{X}_{i,t},\; \mathbf{Y}_{i,t}\big) \in \R^{M+K}.
\end{equation}
The goal is to rank the candidate driver--outcome pairs $\{(x_j, y_k) : j \in [M],\, k \in [K]\}$ by the strength of convergent evidence for a genuine relationship, and to evaluate that ranking against the set $\mathcal{S}$ of true pairs where ground truth is available.

\subsection{Structural-Behavioral Decomposition}\label{sec:decomp}

Let $V = \mathbf{X} \cup \mathbf{Y}$ be all measured variables.

\begin{definition}[Structural Relationship]
A variable $v \in V$ stands in a \emph{structural} relationship to variables $u_1, \ldots, u_m \in V$ if $v = f(u_1, \ldots, u_m)$ holds by definition or by an accounting or measurement identity, independent of any behavioral mechanism. We use ``structural'' in this accounting-identity sense, not in the structural-equation-model sense; the targets of the decomposition are derived-variable identities that are not intervention candidates. Each such identity contributes edges $(u_i, v)$ to the structural graph $G_S \subseteq V \times V$.
\end{definition}

\begin{definition}[Behavioral Relationship]
A relationship $y = g(x_1, \ldots, x_m; \varepsilon)$ is \emph{behavioral} if it represents a mechanism that is invariant under intervention on the $x_i$ \citep{pearl2000} rather than an algebraic re-expression. We denote the behavioral graph $G_B$.
\end{definition}

Examples of structural relationships:
\begin{align}
  \text{Revenue} &= \text{Average Price} \times \text{Volume}, \\
  \text{Yield} &= 1 - \text{Defect Rate}, \\
  \text{ROE} &= \text{Margin} \times \text{Turnover} \times \text{Leverage}.
\end{align}

$G_S$ is known \emph{a priori} from domain knowledge; $G_B$ represents the behavioral relationships whose candidate edges we seek to prioritize. The candidate set after decomposition is
\begin{equation}\label{eq:candidate}
  \mathcal{P} = (\mathbf{X} \times \mathbf{Y}) \setminus \{(u,v) : (u,v) \in G_S \text{ or } (v,u) \in G_S\}.
\end{equation}
Crucially, this screening changes the analysis \emph{only} when an outcome's identity components $u_i$ are themselves candidate drivers (i.e.\ $u_i \in \mathbf{X}$). If the components are not in the feature set, for instance when only the derived outcome is measured, there are no identity pairs to remove and decomposition is a no-op. This is not a limitation so much as a precise statement of scope: the false-positive class that decomposition prevents exists exactly when components and their algebraic combination are both offered to the methods as (driver, outcome) candidates. Section~\ref{sec:exp_decomp} constructs and measures precisely this case. Identity components remain in the graph as ordinary variables; only the tautological component$\rightarrow$composite edges are removed.

\subsection{The Multi-Method Ensemble}

Define a method suite $\mathcal{M} = \{m_1, \ldots, m_{11}\}$ spanning eight distinct mathematical traditions, classical and regularized statistics, energy statistics, panel econometrics, machine learning (including non-parametric single-feature prediction), time-series analysis, information theory, probabilistic graphical models, and causal machine learning. For each method $m_k$ and each candidate pair $(x_j, y_l)$:
\begin{enumerate}[nosep]
  \item $m_k$ produces raw evidence $e_k(x_j, y_l)$ in a method-specific form (correlation coefficient, regression coefficient, SHAP value, p-value, entropy, graph edge, treatment effect).
  \item Normalize: $\tilde{e}_k(x_j, y_l) = \phi_k\big(e_k(x_j, y_l)\big) \in [0, 1]$, where $\phi_k$ is a method-specific normalization function.
  \item Weight: each method receives a nonnegative synthesis weight $w_k$. MCES uses uniform weights by default; alternative weights may encode application-specific evidential preferences and are evaluated by sensitivity analysis (Section~\ref{sec:sensitivity}).
\end{enumerate}

The Convergent Evidence Score is:
\begin{equation}\label{eq:ces}
  \CES(x_j, y_l) = \frac{\sum_{k=1}^{11} w_k \cdot \tilde{e}_k(x_j, y_l)}{\sum_{k=1}^{11} w_k}.
\end{equation}

\subsection{Desiderata}

The framework must satisfy eight requirements:
\begin{enumerate}[nosep]
  \item Adjust for measured covariates, while acknowledging that unmeasured confounding remains unresolved.
  \item Capture linear \emph{and} non-linear relationships.
  \item Distinguish temporal precedence from contemporaneous correlation.
  \item Account for unit-level heterogeneity (different baselines per unit).
  \item Separate structural from behavioral relationships.
  \item Be validatable against known ground truth.
  \item Be domain-agnostic.
  \item Scale to reasonable dimensionality (tens of drivers, handful of outcomes).
\end{enumerate}

\section{The MCES Framework}\label{sec:framework}

This section details each step of the MCES pipeline. Figure~\ref{fig:overview} and Algorithm~\ref{alg:mces} provide an overview.

\begin{figure}[ht]
\centering
\begin{tikzpicture}[
  font=\scriptsize,
  node distance=3.5mm and 5mm,
  box/.style={draw=black!60, rounded corners=1pt, align=center, inner sep=3pt, minimum height=7mm},
  tier/.style={box, fill=black!4, text width=44mm},
  stage/.style={box, fill=blue!6, text width=20mm},
  arr/.style={-{Stealth[length=1.6mm]}, black!70, thick}
]
\node[stage] (data) {Panel data\\ $\mathbf{O} \in \R^{N \times T \times (M+K)}$};
\node[stage, below=6mm of data] (gs) {Structural graph $G_S$ (identities)};
\node[stage, right=5mm of $(data.east)!0.5!(gs.east)$] (decomp) {Structural--Behavioral Decomposition};
\node[tier, right=6mm of decomp, yshift=16mm] (t1) {\textbf{Tier 1 -- Associational}\\ Partial Corr.\ $\cdot$ Lasso $\cdot$ Distance Corr.};
\node[tier, below=2mm of t1] (t2) {\textbf{Tier 2 -- Predictive}\\ Mixed Effects $\cdot$ RF+SHAP $\cdot$ PPS};
\node[tier, below=2mm of t2] (t3) {\textbf{Tier 3 -- Temporal}\\ Granger $\cdot$ ITS $\cdot$ Transfer Entropy};
\node[tier, below=2mm of t3] (t4) {\textbf{Tier 4 -- Structural}\\ Bayesian Net $\cdot$ Causal Forest};
\node[draw=black!40, rounded corners=2pt, fit=(t1)(t4), inner sep=3pt, label={[font=\scriptsize]above:{Method suite (11 methods, 8 traditions)}}] (suite) {};
\node[stage, right=of suite] (pool) {Normalize $\phi_k$ $+$ weighted pool (Eq.~\ref{eq:ces})};
\node[stage, right=of pool] (out) {CES, calibration, convergence bands, ranked pairs};
\draw[arr] (data.east) -- (decomp);
\draw[arr] (gs.east) -- (decomp);
\draw[arr] (decomp) -- node[above=0.5mm, font=\tiny]{$\mathcal{P}$} (suite.west|-decomp);
\draw[arr] (suite.east|-pool) -- (pool);
\draw[arr] (pool) -- (out);
\end{tikzpicture}
\caption{The MCES pipeline. Definitional (identity) pairs are removed first; the surviving behavioral candidates are scored by eleven methods from eight mathematical traditions, normalized, and pooled into the Convergent Evidence Score.}
\label{fig:overview}
\end{figure}
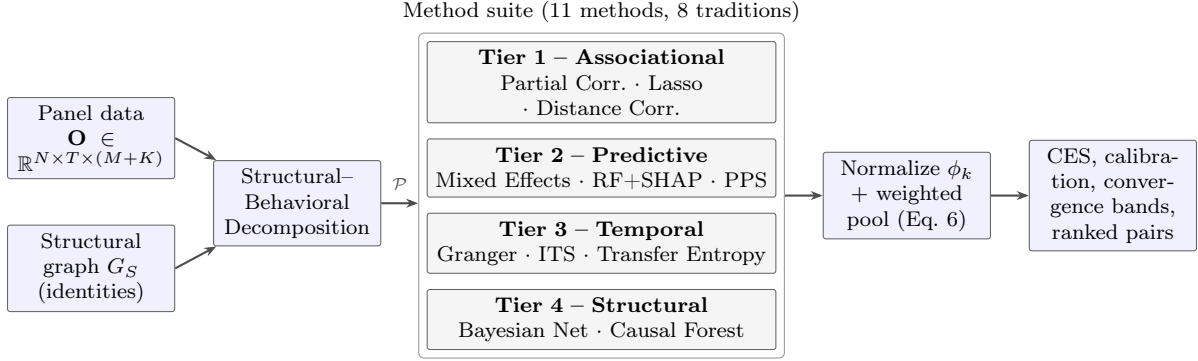

\begin{algorithm}[ht]
\caption{Multi-Method Causal Evidence Synthesis (MCES)}
\label{alg:mces}
\begin{algorithmic}[1]
\Require Observational panel data $\mathbf{O} \in \R^{N \times T \times (M+K)}$, structural graph $G_S$
\Ensure CES matrix $\mathbf{C} \in [0,1]^{M \times K}$
\State $\mathcal{P} \gets \{(x_j, y_k)\} \setminus G_S$ \Comment{Remove structural pairs}
\For{each method $m_k \in \mathcal{M}$}
  \For{each pair $(x_j, y_l) \in \mathcal{P}$}
    \State $e_k(x_j, y_l) \gets m_k(\mathbf{O}, x_j, y_l)$ \Comment{Run method}
    \State $\tilde{e}_k(x_j, y_l) \gets \phi_k(e_k)$ \Comment{Normalize to $[0,1]$}
  \EndFor
\EndFor
\For{each pair $(x_j, y_l) \in \mathcal{P}$}
  \State $\CES(x_j, y_l) \gets \sum_{k} w_k \tilde{e}_k / \sum_{k} w_k$ \Comment{Weighted synthesis}
\EndFor
\State \Return $\mathbf{C}$
\end{algorithmic}
\end{algorithm}

\subsection{Step 1: Structural-Behavioral Decomposition}

Given domain knowledge, we identify all structural (definitional) relationships among outcome variables and remove them from the analysis space. This step requires a domain expert to enumerate algebraic identities. The process is:
\begin{enumerate}[nosep]
  \item Enumerate all outcome variables and their definitions.
  \item For each outcome pair $(y_k, y_l)$: if $y_k = f(y_l, \ldots)$ by definition, add edge $(y_l, y_k)$ to $G_S$.
  \item Remove all pairs involving $G_S$ edges from the candidate set $\mathcal{P}$.
\end{enumerate}

Without this step, methods will ``discover'' that $\text{Revenue} \sim \text{Average Price}$ (trivially true by multiplication) and inflate false positives.

\subsection{Step 2: The Method Suite}\label{sec:methods}

MCES runs eleven methods spanning the mathematical traditions listed above. We detail each method's formulation, unique contribution, and normalization; the implementation defaults (lags, cross-validation folds, discretization bins, tree counts) are those used in all experiments and are documented in the reference implementation.

\subsubsection{Method 1: Partial Correlation}
\emph{Tradition:} Classical statistics.

We control for \emph{all} other measured drivers simultaneously via the precision matrix $\mathbf{P} = \mathbf{R}^{-1}$, where $\mathbf{R}$ is the correlation matrix of $(\mathbf{X}, y_l)$ (regularized as $\mathbf{R} + 10^{-6}\mathbf{I}$ if near-singular). The partial correlation between $x_j$ and $y_l$ given the rest is
\begin{equation}
  r(x_j, y_l \mid \text{rest}) = \frac{-P_{jl}}{\sqrt{P_{jj}\,P_{ll}}},
\end{equation}
with significance from a $t$-test on $n - p$ degrees of freedom.\\
\emph{Unique contribution:} Simplest and most interpretable baseline, measuring conditional linear association after adjustment for the other included variables (not control for unmeasured confounders).\\
\emph{Normalization:} $\tilde{e}_1 = |r| \cdot \mathbb{I}[p < 0.05]$ (absolute, hard-gated).

\subsubsection{Method 2: Lasso Regression}
\emph{Tradition:} Regularized regression.

\begin{equation}
  \hat{\boldsymbol{\beta}} = \argmin_{\boldsymbol{\beta}} \left\{ \frac{1}{2n}\|\mathbf{y} - \mathbf{X}\boldsymbol{\beta}\|_2^2 + \lambda \|\boldsymbol{\beta}\|_1 \right\}.
\end{equation}
with $\lambda$ chosen by 5-fold cross-validation on standardized drivers.\\
\emph{Unique contribution:} Automatic elimination of irrelevant variables via the $\ell_1$ penalty.\\
\emph{Normalization:} $\tilde{e}_2 = |\hat{\beta}_j| / \max_{j'} |\hat{\beta}_{j'}|$ over drivers of the same outcome; $0$ for eliminated features. This is a \emph{within-outcome relative} scale (see Section~\ref{sec:norm}).

\subsubsection{Method 3: Distance Correlation}
\emph{Tradition:} Energy statistics.

Distance correlation \citep{szekely2007} measures dependence of arbitrary form. With $A$ and $B$ the double-centered pairwise-distance matrices of the samples of $x_j$ and $y_l$,
\begin{equation}
  \text{dCor}(x_j, y_l) = \frac{\text{dCov}(x_j, y_l)}{\sqrt{\text{dVar}(x_j)\,\text{dVar}(y_l)}}, \qquad \text{dCov}^2 = \frac{1}{n^2}\sum_{a,b} A_{ab} B_{ab},
\end{equation}
and $\text{dCor} = 0$ \emph{if and only if} $x_j$ and $y_l$ are independent, unlike Pearson-family measures, which can be zero under symmetric (e.g.\ U-shaped) dependence. Significance is assessed by a permutation test on $\text{dCov}$.\\
\emph{Unique contribution:} A Tier-1 associational measure that does not assume a specific functional form and can detect non-monotone dependence, subject to finite-sample power; it is the pool's designated detector of symmetric non-linear dependence, complementing the monotone-oriented classical measures.\\
\emph{Normalization:} $\tilde{e}_3 = \text{dCor} \cdot \mathbb{I}[p_{\text{perm}} < 0.05]$ ($\text{dCor}$ is already in $[0,1]$). Distance correlation is \emph{unsigned} by construction, and the non-monotone dependence it is designed to catch has no well-defined sign; we therefore treat its directional contribution as undefined and record a Pearson-sign annotation only as an interpretive hint valid under monotonicity, not as directional evidence.

\subsubsection{Method 4: Mixed-Effects Regression}
\emph{Tradition:} Panel econometrics.

\begin{equation}
  y_{it} = \mathbf{X}_{it}^\top \boldsymbol{\beta} + b_i + \varepsilon_{it}, \quad b_i \sim \mathcal{N}(0, \sigma_b^2),
\end{equation}
fitted by REML with a random \emph{intercept} $b_i$ per unit (shared slopes $\boldsymbol{\beta}$); an OLS fallback is used if the mixed model fails to converge.\\
\emph{Unique contribution:} Handles repeated measures from the same unit with unit-specific baselines.\\
\emph{Normalization:} $\tilde{e}_4 = \big(|\hat{\beta}_j| / \max_{j'} |\hat{\beta}_{j'}|\big) \cdot \mathbb{I}[p_j < 0.05]$ (within-outcome relative, hard-gated).

\subsubsection{Method 5: Random Forest + SHAP}
\emph{Tradition:} Machine learning with explainability.

A Random Forest $f$ is fitted to predict $y_l$ from $\mathbf{X}$. SHAP values \citep{lundberg2017} decompose each prediction:
\begin{equation}
  f(\mathbf{x}_i) = \phi_0 + \sum_{j=1}^{M} \phi_j(\mathbf{x}_i),
\end{equation}
where $\phi_j(\mathbf{x}_i)$ is the Shapley value of feature $j$ for observation $i$.\\
\emph{Unique contribution:} Captures non-linear relationships and feature interactions with no linearity assumption.\\
\emph{Normalization:} $\tilde{e}_5 = \overline{|\phi_j|} / \max_{j'} \overline{|\phi_{j'}|}$ (within-outcome relative). The magnitude is a mean absolute SHAP value and is inherently unsigned; where we report a direction we use the sign of the mean signed contribution $\overline{\phi_j}$, which is meaningful only for approximately monotone effects and is not defined for interaction-dominated or non-monotone importance.

\subsubsection{Method 6: Predictive Power Score}%
\label{sec:method_pps}
\emph{Tradition:} Machine learning (non-parametric single-feature prediction).

The Predictive Power Score \citep{wetschoreck2020} asks how well the single driver $x_j$ predicts $y_l$ out-of-sample. A shallow decision tree is fitted to predict $y_l$ from $x_j$ alone and evaluated by $K$-fold cross-validation against a naive median baseline:
\begin{equation}
  \text{PPS}(x_j \to y_l) = \max\left(0,\; 1 - \frac{\text{MAE}_{\text{tree}}}{\text{MAE}_{\text{median}}}\right).
\end{equation}
Unlike the symmetric association measures, PPS is \emph{asymmetric}: $\text{PPS}(x \to y)$ generally differs from $\text{PPS}(y \to x)$ (predicting $y=|x|$ from $x$ is easy; recovering $x$ from $y$ is not). We stress that this is \emph{predictive} asymmetry, not causal direction: it reflects the deterministic structure of the map, and a confounded or reverse-causal pair can be equally asymmetric, so we do not treat PPS asymmetry as evidence of which variable causes which.\\
\emph{Unique contribution:} Out-of-fold single-feature predictive skill at negligible cost; complements the multivariate methods, which can dilute a strong marginal predictor.\\
\emph{Normalization:} $\tilde{e}_6 = \text{PPS} \in [0,1]$ (already baseline-normalized; scores at or below the naive baseline are exactly $0$).

\subsubsection{Method 7: Granger Causality}
\emph{Tradition:} Time series econometrics.

Variable $x_j$ Granger-causes $y_l$ if past values of $x_j$ improve the prediction of $y_l$ beyond $y_l$'s own history:
\begin{equation}
  y_{l,t} = \sum_{\tau=1}^{p} \alpha_\tau y_{l,t-\tau} + \sum_{\tau=1}^{p} \gamma_\tau x_{j,t-\tau} + \varepsilon_t.
\end{equation}
An F-test assesses whether the $\gamma_\tau$ coefficients are jointly significant. We run the test per unit over lags $1,\ldots,3$ and pool the per-unit $p$-values across units by Fisher's method, $\chi^2 = -2\sum_i \ln p_i$. Taking the \emph{best} lag per unit inflates significance through selection, so we treat the resulting $p$-value as a screening statistic rather than a calibrated test; a joint test over all lags, or a multiplicity correction, would be more conservative and is preferable when a calibrated $p$-value is needed. Fisher pooling additionally assumes independence across units, which common shocks can violate; this is a known conservative-to-anticonservative trade-off we flag in the limitations.\\
\emph{Unique contribution:} Tests whether a driver's past values provide incremental predictive information about the outcome beyond the outcome's own past (predictive precedence, not proven causation).\\
\emph{Normalization:} $\tilde{e}_7 = (1 - p_{\text{pool}}) \cdot \mathbb{I}[p_{\text{pool}} < 0.05]$. Because $1-p$ conflates effect size with sample size, we interpret $\tilde{e}_7$ as evidence \emph{that} a temporal effect is nonzero rather than its magnitude.

\subsubsection{Method 8: Interrupted Time Series}
\emph{Tradition:} Quasi-experimental design.

Classical ITS regresses an outcome on time, a post-intervention indicator $D_t = \mathbb{I}[t \geq t_0]$, and their interaction. Because in our setting the ``treatment'' is a continuous driver rather than a single system-wide event, we use a \emph{driver-moderated} ITS: with $\tilde{x}_j$ the standardized driver,
\begin{equation}
  y_t = \beta_0 + \beta_1 t + \beta_2 D_t + \beta_3 (t-t_0)D_t + \beta_4 \tilde{x}_{j,t} + \beta_5 (\tilde{x}_{j,t} \cdot D_t) + \varepsilon_t,
\end{equation}
and score the interaction $\beta_5$: the change in the driver's association with the outcome after the break.\\
\emph{Unique contribution:} Leverages known structural breaks (e.g.\ software deployments, policy changes).\\
\emph{Normalization:} $\tilde{e}_8 = \big(|\hat{\beta}_5| / \max_{j'}|\hat{\beta}_5^{(j')}|\big)\cdot \mathbb{I}[p_{\beta_5} < 0.05]$ (within-outcome relative, hard-gated).\\
\emph{Note:} A genuine interrupted time series requires a substantively justified breakpoint. When a break point $t_0$ is declared (e.g.\ a deployment week) this method is a true ITS. When none is declared, the mid-window default $t_0 = \lfloor T/2 \rfloor$ makes it \emph{not} an intervention analysis but a \emph{structural-change diagnostic} (a test for a shift in the driver--outcome relationship across the observation window), and we interpret and label it as such. Absent a real breakpoint we do not read its output as an intervention effect; the true break should be supplied when known.

\subsubsection{Method 9: Transfer Entropy}
\emph{Tradition:} Information theory.

The transfer entropy from $x_j$ to $y_l$ measures directed information flow:
\begin{equation}
  \text{TE}_{x_j \to y_l} = \sum p\big(y_{t+1}, y_{t}^{(\kappa)}, x_{t}^{(\kappa)}\big) \log \frac{p\big(y_{t+1} \mid y_{t}^{(\kappa)}, x_{t}^{(\kappa)}\big)}{p\big(y_{t+1} \mid y_{t}^{(\kappa)}\big)},
\end{equation}
where superscript $(\kappa)$ denotes history length (we use $\kappa = 1$). We estimate the entropies with a plug-in estimator over $5$ quantile bins and obtain $p_{\text{perm}}$ from $50$ within-unit temporal permutations of $x_j$.\\
\emph{Unique contribution:} Captures non-linear directed information flow. \citet{barnett2009} proved equivalence with Granger causality \emph{only} for Gaussian variables; for non-Gaussian data, transfer entropy detects relationships Granger misses. (Conversely, on near-Gaussian data the two are redundant, a caveat for weighting, addressed in Section~\ref{sec:diversity}.)\\
\emph{Normalization:} $\tilde{e}_9 = (\text{TE} / \max_{j'} \text{TE}^{(j')}) \cdot \mathbb{I}[p_{\text{perm}} < 0.05]$. The plug-in estimator is biased upward in small samples; the permutation gate mitigates this by testing against a shuffled null. The permutation $p$-value uses the finite-sample form $(b+1)/(B+1)$, so with $B = 50$ its resolution is coarse (minimum attainable $p \approx 0.02$); transfer entropy's gate is therefore blunter than the analytic gates of other methods, and increasing $B$ is a straightforward refinement at additional compute cost.

\subsubsection{Method 10: Bayesian Network Structure Learning}
\emph{Tradition:} Probabilistic graphical models.

Using BIC-scored hill-climbing \citep{pearl2000, spirtes2000} over data discretized into $6$ bins, with drivers pre-selected to a top-$20$ multi-signal shortlist for tractability (the union of the top-$k$ drivers by maximum absolute correlation with any outcome, drivers with nonzero lasso coefficients, and the top-$k$ by mutual information, truncated to $20$ by a combined normalized score) and $\max$ in-degree $3$, we learn a high-scoring DAG $\hat{G}$:
\begin{equation}
  \hat{G} = \argmax_{G \in \mathcal{G}} \text{BIC}(\mathbf{O} \mid G).
\end{equation}
\emph{Unique contribution:} Provides a candidate graphical representation of adjacency and short directed paths (subject to Markov-equivalence non-identifiability, below).\\
\emph{Normalization:} $\tilde{e}_{10} = 1$ if a direct edge $x_j \to y_l$ exists in $\hat{G}$; $0.5$ if a directed path of length $2$ exists; $0$ otherwise. Hill-climbing returns a single member of a Markov equivalence class, so edge \emph{orientation} is only partially identified; the score should be read as adjacency-plus-orientation-under-the-learned-DAG, not as identified direction. Drivers outside the pre-selected shortlist are recorded as \emph{not evaluated} rather than as zero evidence (Section~\ref{sec:synthesis}).

\subsubsection{Method 11: Causal Forest}
\emph{Tradition:} Causal machine learning.

Using the framework of \citet{wageracthey2018}, we estimate heterogeneous treatment effects:
\begin{equation}
  \hat{\tau}(x) = \E[Y(1) - Y(0) \mid X = x],
\end{equation}
where $Y(1)$ and $Y(0)$ are potential outcomes under treatment and control.\\
We binarize each continuous driver at its median to form the treatment and use the remaining pre-selected drivers as controls; the reported statistic is the average treatment effect $\hat{\tau} = \E[\hat{\tau}(x)]$ (an econML \texttt{CausalForestDML}, with a $t$-test fallback).\\
\emph{Unique contribution:} \texttt{CausalForestDML} permits heterogeneous response surfaces across units, although the CES contribution we pool is the estimated average high-versus-low treatment contrast $\hat\tau$ rather than a direct measure of heterogeneity.\\
\emph{Normalization:} $\tilde{e}_{11} = (|\hat{\tau}| / \max_{j'} |\hat{\tau}^{(j')}|) \cdot \mathbb{I}[\text{95\% CI excludes } 0]$. Two cautions apply. Median binarization discards dose information, so $\hat{\tau}$ is a coarse high-vs-low contrast (drivers outside the top-$25$ shortlist are recorded as not evaluated). More importantly, treating each driver in turn as the treatment while using all remaining drivers as controls can, without a stated causal graph, condition on mediators, colliders, or post-treatment variables and thereby bias $\hat{\tau}$; we therefore read $\tilde{e}_{11}$ as one heterogeneity-sensitive evidence signal among many, not as an identified average treatment effect. Supplying a causal graph to choose valid adjustment sets is the principled fix and is left to future work.

\subsection{Step 3: Evidence Normalization}\label{sec:norm}

Each method produces outputs in different scales and types: correlation coefficients in $[-1,1]$, regression coefficients in $\R$, SHAP values in $\R$, $p$-values in $[0,1]$, entropy values in $\R^+$, binary graph edges, and treatment effects in $\R$. These are fundamentally \emph{different analytical quantities and evidence measures}, SHAP measures predictive attribution, Granger measures temporal predictability, Bayesian networks measure graphical dependence, and causal forests measure treatment effects. Normalization maps each to $[0,1]$, but this numerical commensurability does not imply semantic equivalence.

The normalized score $\tilde{e}_k(x_j, y_l)$ should therefore be interpreted as: ``how strongly does method $m_k$, through its particular analytical lens, indicate that driver $x_j$ is relevant to outcome $y_l$?'' CES then measures \emph{convergence of evidence} across these lenses, not the magnitude of any particular causal effect.

Two properties of our $\phi_k$ deserve emphasis because they qualify how CES may be read:
\begin{itemize}[nosep]
  \item \textbf{Within-outcome relative scaling.} Six of the eleven methods (lasso, mixed effects, RF+SHAP, ITS, transfer entropy, causal forest) divide by the maximum statistic \emph{across drivers of the same outcome}; distance correlation and PPS are natively $[0,1]$-scaled and are the exceptions among the associational/predictive tiers. Consequently $\tilde{e}_k$ ranks drivers reliably \emph{within} an outcome but is not comparable in absolute terms \emph{across} outcomes: an outcome whose strongest driver is weak still awards $\tilde{e}_k \approx 1$ to that driver. CES rankings and all our metrics are computed per-outcome-aware (Precision@$K$ over the pooled grid still holds because true edges score high within their own outcome), but practitioners should treat raw CES as an ordinal, within-outcome quantity unless calibrated (Section~\ref{sec:calibration}).
  \item \textbf{Significance and selection gates.} Several methods apply a significance or confidence-interval gate (an indicator $\mathbb{I}[p<0.05]$, or ``CI excludes $0$'', or edge-present). Methods without inferential $p$-values instead use their method-specific criteria: sparse selection (lasso), relative importance (SHAP), out-of-fold predictive performance (PPS), or graph presence (Bayesian network). Where a gate applies it enforces ``zero evidence when not significant'' but introduces a discontinuity at the threshold; the smoothness of the linear pool (Section~\ref{sec:theory}) therefore holds \emph{above} such gates, not through them. For methods that produce inferential $p$-values, MCES applies Benjamini--Hochberg adjustment \citep{benjamini1995} across the driver--outcome grid by default (Section~\ref{sec:exp_fp}).
\end{itemize}

Each $\phi_k$ is bounded in $[0,1]$ and non-decreasing in its underlying statistic; the gated methods return zero when their significance or selection criterion is not met.

\subsection{Step 4: Weighted Synthesis (CES Computation)}\label{sec:weights}\label{sec:synthesis}

Under the default, every method receives the same weight. For the optional tiered scheme, methods are grouped into four analytical tiers (Table~\ref{tab:weights}) that encode a contestable evidential preference for temporal and structural methods; this grouping affects only the optional scheme and, as Section~\ref{sec:sensitivity} shows, has no measurable effect on accuracy.

\begin{table}[ht]
\centering
\caption{Analytical groupings and optional tiered weights. The default is uniform; the tiered column is one alternative studied in the weight sensitivity analysis.}
\label{tab:weights}
\small
\begin{tabular}{@{}P{3.0cm}P{4.5cm}P{4.3cm}l@{}}
\toprule
\textbf{Tier} & \textbf{Methods} & \textbf{What They Provide} & \textbf{$w_k$} \\
\midrule
1: Associational & Partial Corr., Lasso, Distance Corr. & Associations, no direction & 0.06 \\
2: Predictive & Mixed Effects, RF+SHAP, PPS & Importance, structure & 0.05--0.08 \\
3: Temporal & Granger, ITS, Transfer Entropy & Direction and timing & 0.13--0.14 \\
4: Structural and treatment-effect & Bayesian Net, Causal Forest & Causal structure, effects & 0.10--0.11 \\
\bottomrule
\end{tabular}
\end{table}

The optional tiered weight vector (summing to 1.0, in method order) is:
\begin{equation}\label{eq:weights}
\mathbf{w} = (0.06,\; 0.06,\; 0.06,\; 0.07,\; 0.08,\; 0.05,\; 0.14,\; 0.14,\; 0.13,\; 0.11,\; 0.10).
\end{equation}

\paragraph{Uniform weights are the default; tiered weights are optional.} We recommend \emph{uniform} weights ($w_k = 1/11$) as the primary default, and we report all headline results under a weighting to which, as Section~\ref{sec:sensitivity} shows, the CES ranking is nearly invariant (mean pairwise Spearman \weightSpearman\ across schemes). The tiered vector in Equation~\ref{eq:weights} encodes a \emph{contestable} epistemic prior, that temporal and treatment-effect methods deserve more credence than purely associational ones. We are explicit that this prior is not established: temporal precedence, in particular, is not stronger identification than a credible treatment-effect design, so we do not order Tier~3 above Tier~4 on identification grounds, and Section~\ref{sec:sensitivity} shows the tiering yields no measurable accuracy gain over uniform. When all eleven methods evaluate a pair, the optional tiered vector has the interpretive property that Moderate-or-higher CES requires positive Tier~3 or Tier~4 evidence. This property is not guaranteed after renormalization in degraded modes. The framework therefore relies only on the explicit Tier~3/4 gate for Strong convergence (Section~\ref{sec:convclass}), which applies under every weighting scheme. Practitioners should use the uniform default; the tiered scheme is offered only for the sensitivity analysis.

The CES is computed via Equation~\ref{eq:ces}. By construction, $\CES \in [0, 1]$.

\paragraph{Missing methods: renormalization at the method level, zero-fill at the pair level.} Two kinds of missingness are handled differently, and the distinction matters for comparability. When a method is inapplicable to the dataset as a whole (e.g.\ ITS with no declared break on cross-sectional data), its weight is removed and the remaining weights renormalize, so degraded modes remain well-defined. When an applicable method skips an individual pair (e.g.\ a driver outside the BN or causal-forest pre-selection shortlist), that pair receives \emph{zero evidence from that method over the unchanged denominator}, a deliberately conservative choice: not being shortlisted lowers CES rather than inflating it by shrinking the denominator, and because the denominator is fixed given the active method set, CES is the same function on every pair; per-pair denominators never vary. Since shortlist membership is itself signal-dependent, this zero-fill is the safe direction for the missing-not-at-random concern. An ablation confirms the design is not load-bearing on the primary scenario, where all eleven methods are active: recomputing CES under the alternative convention yields identical rankings (Spearman \renormSpearman\ across 20 seeds, F1@10 \renormFOne\ vs.\ \zeroFillFOne). We track the count of evaluating methods separately from the count contributing positive evidence.

\paragraph{A weight-vector property, backed by an explicit gate.} Under the default uniform weights with all eleven methods evaluating a pair, the six Tier~1--2 (associational and predictive) methods carry mass $6/11 = 0.545 < 0.7$, so $\CES > 0.7$ cannot be reached on associational and predictive evidence alone. The optional tiered weights (Equation~\ref{eq:weights}) strengthen this: because their Tier~1--2 mass is only $0.38 < 0.4$, under tiered weights even \emph{Moderate} convergence ($\CES \ge 0.4$) would require a Tier~3/4 method, the sole thing tiering buys over uniform (Section~\ref{sec:sensitivity}). We note plainly that this is a property \emph{by construction}, not a finding: the tier constants are design parameters we chose, and a Tier~1--2 mass below the Moderate threshold is a direct consequence of that choice. We do \emph{not}, however, rely on the weight vector for this guarantee, because renormalization in degraded modes can break it (Section~\ref{sec:convclass} gives the cross-sectional counterexample, where seven methods apply and Tier~1--2 mass reaches $5/7 = 0.714 > 0.7$). The Tier~3/4 requirement for Strong convergence is therefore enforced as an explicit rule in the classifier (Section~\ref{sec:convclass}), which holds in every mode.

\paragraph{Applicability conditions and degraded modes.} Not every dataset supports every method, and the renormalization above makes degradation graceful rather than fatal. Cross-sectional data (no time dimension) disables the three temporal methods, as on the Sachs benchmark. A \emph{single-unit} time series ($N=1$), the common case of one organization observed daily, disables the panel-dependent mixed-effects method; the remaining ten methods, including all temporal ones, still apply, and weights renormalize over the active pool. The reference implementation detects these conditions and skips inapplicable methods automatically. Because the explicit Strong-convergence gate (Section~\ref{sec:convclass}) requires an applicable Tier~3/4 method regardless of the renormalized weights, the directional requirement is preserved in every degraded mode, including cross-sectional data, where the weight-vector argument alone would fail ($5/7 = 0.714 > 0.7$).

\subsection{Step 5: Convergence Classification}\label{sec:convclass}

We bin each driver-outcome pair into three \emph{convergence} levels. We use ``convergence'' rather than ``confidence'' deliberately: the bands describe how strongly the heterogeneous methods agree, not a probability of causation (Section~\ref{sec:norm} notes that raw CES is ordinal and within-outcome relative).
\begin{itemize}[nosep]
  \item \textbf{Strong Convergence}: $\CES > 0.7$ \emph{and} at least one applicable Tier~3 or Tier~4 method contributes positive evidence.
  \item \textbf{Moderate Convergence} ($0.4 \leq \CES \leq 0.7$): Some methods agree, mixed signals across tiers.
  \item \textbf{Weak Convergence} ($\CES < 0.4$): Weak or inconsistent evidence.
\end{itemize}

The $0.4$ and $0.7$ cutoffs are interpretive design thresholds, not calibrated probabilities or universally validated decision boundaries; their suitability should be evaluated for each target application. We enforce the Tier~3/4 requirement for Strong convergence as an \emph{explicit rule} in the classifier, not as a consequence of the weight vector. When all eleven methods evaluate a pair under uniform weights, the requirement also follows automatically because the Tier~1--2 mass is $6/11 = 0.545 < 0.7$; but under renormalization in degraded modes (e.g.\ cross-sectional data, where only seven methods apply and Tier~1--2 mass can reach $5/7 = 0.714 > 0.7$) the weight-vector argument alone would not hold, so the explicit gate is what preserves the intended interpretation across all modes. We also compute an inter-method \emph{agreement} metric as the fraction of \emph{applicable} methods producing moderate-or-higher evidence:
\begin{equation}
  \text{Agreement}(x_j, y_l) = \frac{1}{|\mathcal{M}_{\text{eval}}(x_j,y_l)|}\sum_{k \in \mathcal{M}_{\text{eval}}(x_j,y_l)} \mathbb{I}[\tilde{e}_k(x_j, y_l) > 0.5],
\end{equation}
where $\mathcal{M}_{\text{eval}}$ is the set of methods that evaluated the pair (so a driver dropped by a pre-selection step does not deflate the denominator).

\subsection{Scope note: decision rules are out of scope}\label{sec:scope_decision}

CES scores and ranks the \emph{evidence} for driver-outcome relationships. Turning that ranking into an action ordering requires external decision criteria (intervention feasibility, cost, risk, expected utility) that are domain-specific and not part of the methodological contribution we validate here. We therefore do not fold such criteria into CES; combining evidence convergence with a validated decision rule is deferred to future work (Section~\ref{sec:futurework}).

\section{Theoretical Justification}\label{sec:theory}

\subsection{MCES as a Committee of Diverse Experts}

We frame MCES using the lens of ensemble learning theory and the ``wisdom of crowds'' literature. Each method $m_k$ acts as an imperfect expert with its own biases (systematic errors from violated assumptions) and variance (sensitivity to noise). The key insight is not that these experts are \emph{independent} (they share the same data) but that they are \emph{diverse}: their errors arise from different mathematical assumptions.

\begin{definition}[Assumption Diversity]
Two methods $m_j$ and $m_k$ are \emph{assumption-diverse} if their failure modes are driven by different violated conditions. For example, Granger causality fails when relationships are non-linear (stationarity/linearity assumption), while SHAP fails when associations are non-causal (no causal identification assumption). A non-linear causal relationship will cause Granger to miss what SHAP detects; a spurious temporal correlation will cause Granger to flag what Causal Forest rejects.
\end{definition}

This diversity is the source of MCES's strength. When methods from different traditions agree, the agreement is informative precisely \emph{because} the methods fail in different ways.

\subsection{Why Pooling Could Help: A Score-Stability Argument}\label{sec:variance}

The argument in this subsection is a conditional, theoretical one: it states when pooling reduces the variance of the score, given low cross-method correlations. Section~\ref{sec:diversity} measures the proposition's own correlation quantity directly, across repeated draws of the data-generating process for fixed pairs, and finds it low ($\bar\rho \approx 0.13$); the step the theory does not supply, and the experiments do not automatically deliver, is from reduced score variance to improved ranking accuracy.

We analyze the estimator MCES actually uses, the weighted pool of Equation~\ref{eq:ces}, rather than a unanimity vote. Fix an outcome $y_l$ and a candidate driver $x_j$. Let $S = \sum_k w_k \tilde{e}_k$ with $\sum_k w_k = 1$ be the (renormalized) CES score, and treat each $\tilde{e}_k \in [0,1]$ as a random variable over resamples of the data-generating process. Write $\mu_k = \E[\tilde{e}_k]$, $\sigma_k^2 = \mathrm{Var}(\tilde{e}_k)$, and let $\rho_{jk}$ be the correlation between $\tilde{e}_j$ and $\tilde{e}_k$. Then
\begin{equation}\label{eq:poolvar}
  \E[S] = \sum_k w_k \mu_k, \qquad
  \mathrm{Var}(S) = \sum_k w_k^2 \sigma_k^2 + \sum_{j \neq k} w_j w_k \rho_{jk}\, \sigma_j \sigma_k .
\end{equation}

The mean of the pool is the weighted mean of the individual signals; the \emph{variance}, however, depends on the cross-method correlations $\rho_{jk}$. This is the crux of the diversity argument, stated for the real estimator:

\begin{proposition}[Diversity reduces score variance]\label{prop:main}
$\mathrm{Var}(S)$ is non-decreasing in every $\rho_{jk}$. With equal weights $w_k = 1/K$ and comparable per-method variances $\sigma_k^2 \approx \sigma^2$,
\begin{equation}
  \mathrm{Var}(S) \approx \frac{\sigma^2}{K}\big(1 + (K-1)\bar{\rho}\big),
\end{equation}
where $\bar{\rho}$ is the mean pairwise correlation. As $\bar{\rho} \to 1$ (redundant methods) the variance tends to $\sigma^2$, no better than one method; as $\bar{\rho} \to 0$ (assumption-diverse methods) it tends to $\sigma^2/K$.
\end{proposition}

\begin{proof}
$\partial\,\mathrm{Var}(S)/\partial \rho_{jk} = 2 w_j w_k \sigma_j \sigma_k \geq 0$ gives monotonicity; substituting $w_k = 1/K$, $\sigma_k = \sigma$, and $\rho_{jk} = \bar{\rho}$ ($j\neq k$) into Equation~\ref{eq:poolvar} gives the stated form.
\end{proof}

\begin{corollary}[Sharper separation of causal from null pairs]\label{cor:sep}
Suppose true pairs have mean pooled score $\mu_1$ and null pairs $\mu_0 < \mu_1$, with per-class pooled standard deviations $s_1, s_0$ that decrease as $\bar\rho$ decreases (Proposition~\ref{prop:main}). Then the standardized separation $(\mu_1 - \mu_0)/\sqrt{\tfrac12(s_1^2+s_0^2)}$ \emph{increases} as diversity increases, so, when the class-conditional score distributions otherwise remain comparable, class separation and the achievable precision/recall trade-off improve. The means $\mu_0,\mu_1$ are set by the methods' individual power and are unchanged by pooling; diversity buys its advantage through variance, not through inflating the signal.
\end{corollary}

This is a claim about \emph{ranking stability and separation}, not a guarantee that CES dominates every individual method on every dataset. Whether the ensemble's lower-variance score actually out-ranks the single best method depends on how much signal ($\mu_1-\mu_0$) the diverse-but-weaker methods contribute versus the noise they add, an empirical question we examine directly, including cases where the ensemble does \emph{not} win, in Section~\ref{sec:results}. Section~\ref{sec:diversity} provides a descriptive measure of how differently the methods rank candidate pairs; it does not directly estimate the resample-level correlations used in Proposition~\ref{prop:main}, and it discusses when redundancy (e.g.\ Granger and transfer entropy on Gaussian temporal data) would erode the benefit.

\subsection{What CES Measures (and Does Not Measure)}

It is essential to state precisely what CES quantifies:

\begin{itemize}[nosep]
  \item \textbf{CES measures:} the degree to which multiple analytical methods with different mathematical assumptions converge on the same driver-outcome relationship. High CES means that multiple methods spanning different analytical traditions point toward the same pair.
  \item \textbf{CES does not measure:} causal identification in the interventionist sense \citep{pearl2000}. CES does not establish that $\text{do}(x_j = x')$ would change $y_l$. Definitive causal claims require controlled experiments or instruments that MCES does not assume.
  \item \textbf{CES approximates:} the strength of \emph{convergent evidence} for \emph{causal relevance}, a pragmatic assessment that a driver-outcome relationship is worth prioritizing for investigation, not a probability that intervening on the driver would change the outcome.
\end{itemize}

This positioning is analogous to how meta-analysis provides ``strength of evidence'' rather than definitive proof: agreement across independent analyses increases confidence, but cannot eliminate all sources of bias.

\subsection{When MCES Fails: Shared Failure Modes}\label{sec:failures}

Assumption diversity protects against method-specific blind spots, but does not protect against \emph{shared} failure modes. We identify three scenarios where all methods can agree incorrectly:

\begin{enumerate}[nosep]
  \item \textbf{Hidden confounder.} An unmeasured variable $Z$ that causes both $x_j$ and $y_l$ will induce a spurious association that all methods detect. Partial correlation controls for \emph{measured} confounders; it cannot control for unmeasured ones. This is a fundamental limitation of all observational methods, not specific to MCES.
  \item \textbf{Feedback loops.} When $x_j$ causes $y_l$ and $y_l$ simultaneously causes $x_j$ (contemporaneous feedback), methods may incorrectly estimate direction. Granger causality and transfer entropy can partially detect bidirectional flow, but contemporaneous feedback remains challenging.
  \item \textbf{Collider bias.} Conditioning on a common effect of $x_j$ and $y_l$ can create a spurious association. If the conditioning variable is included in the analysis, all methods may report a false relationship.
\end{enumerate}

These failure modes are not unique to MCES, they affect every observational causal method. MCES reduces method-specific errors through assumption diversity but cannot eliminate data-level bias: if the data itself is confounded, all methods will reflect that confounding. MCES's advantage is that for the more common scenario where individual methods fail due to their \emph{specific} assumptions (linearity, stationarity, parametric form), diversity can reduce sensitivity to those method-specific failure modes and improve the stability of the pooled score. MCES should therefore be viewed as prioritizing hypotheses for intervention rather than replacing experimental causal identification.

\subsection{Comparison to Meta-Analysis and Ensemble Learning}

Meta-analysis combines \emph{results} from different \emph{studies} of the same question. MCES combines \emph{methods} on the same \emph{data} for the same question. Both leverage the principle that agreement across diverse analyses increases confidence.

The connection to ensemble learning is also instructive. Random forests aggregate decision trees that are diverse due to feature subsampling; boosting aggregates weak learners that focus on different error regions. MCES aggregates analytical methods that are diverse due to fundamentally different mathematical assumptions. The mechanism is analogous: diversity of errors can reduce sensitivity to method-specific failure modes and improve the stability of the pooled score, even when components share the same underlying data.

Formally, CES (Equation~\ref{eq:ces}) is an instance of \emph{linear opinion pooling}, a well-studied aggregation rule \citep{stone1961, genest1986}: $K$ experts provide assessments $p_1,\ldots,p_K$ combined via $p = \sum_k w_k p_k$, $\sum_k w_k = 1$. Linear pooling uniquely satisfies the unanimity-preservation and marginalization properties among a broad class of rules. Two caveats apply to our use of it. First, the $\tilde{e}_k$ are normalized evidence scores, not probabilities, so the pool inherits their ordinal, within-outcome character (Section~\ref{sec:norm}); the isotonic mapping of Section~\ref{sec:calibration} has an empirical probability interpretation only within the calibration distribution, and its validity for a new target domain requires separate evaluation. Second, while the pool is smooth in the $\tilde{e}_k$, each $\tilde{e}_k$ itself contains a hard significance gate, so CES is \emph{not} globally continuous in the underlying statistics, the smoothness holds above the gates. This differs from Super Learner \citep{vanderlaan2007}, which learns weights by cross-validated minimization of a single predictive loss for one estimand; our weights are fixed epistemic priors over methods targeting different estimands, and Section~\ref{sec:sensitivity} shows the ranking is insensitive to them.

\section{Experimental Setup}\label{sec:validation}

All numbers reported in this section are produced by our reference implementation; the exact scripts that generate every table are deterministic given the reported seeds, and are available from the authors on request. Method results are cached so the full suite is reproducible.

\subsection{Synthetic Ground Truth Design}

We generate observational panel data with \textbf{known} embedded causal relationships. A scenario fixes $N$ units, $T$ periods, a set of true causal edges (driver, outcome, sign, strength, lag, functional form), and a data-generating process (DGP). Five DGPs are implemented so that recovery is not tested under a single functional form: linear, non-linear (quadratic/$\sqrt{}$/log/threshold), confounded (hidden common causes), time-lagged, and mixed. Ground-truth edge strengths lie in $[0.15, 0.6]$; unit-specific intercepts and Gaussian noise are added. The primary scenario (\texttt{primary\_panel}) has $N=23$, $T=20$, $95$ numeric candidate drivers, $6$ outcomes, and $18$ true edges. We additionally use \texttt{nonlinear} ($N=50, T=30$) and \texttt{confounded} ($N=30, T=20$) scenarios, and two further domains described below.

\subsection{Benchmarks and Additional Domains}

To reduce the circularity of validating only on self-designed synthetic data, we add an external real-data benchmark, a suite of named structure-learning benchmarks, and two independently-specified synthetic domains:

\begin{itemize}[nosep]
  \item \textbf{Sachs protein-signaling dataset} \citep{sachs2005}, \emph{real} flow-cytometry measurements of $11$ phosphoproteins ($853$ observational cells), with the widely-used consensus network as ground truth. It is cross-sectional, so only the seven non-temporal, non-panel methods apply; under the driver/outcome orientation we use, $16$ of the $17$ consensus edges are scoreable. This is the one non-synthetic dataset in our evaluation, and its DGP was certainly not designed with our methods in mind.
  \item \textbf{Bayesian-network structure benchmarks}, six standard \texttt{bnlearn} networks \citep{scutari2010} spanning distinct domains and $20$ to $76$ nodes: \textsc{Child} ($20$, congenital heart disease), \textsc{Insurance} ($27$, actuarial risk), \textsc{Alarm} ($37$, ICU monitoring), \textsc{Hailfinder} ($56$, severe-weather forecasting), \textsc{Hepar2} ($70$, hepatology), and \textsc{Win95pts} ($76$, fault diagnosis), each with an \emph{exact published DAG}. We forward-sample $1{,}000$ observations, ordinal-encode the discrete states, and (as with Sachs) assign each node a single driver/outcome role by net edge direction, so the scoreable ground truth is the set of forward edges. These test edge recovery across a ladder of network sizes and across six unrelated domains, with structure ground truth entirely external to our framework. The data are sampled rather than field-collected; the \emph{structure} is the real, published, citable object.
  \item \textbf{Healthcare} ($40$ hospitals $\times$ $24$ months, mixed DGP with hidden confounders, $10$ true edges) and \textbf{Manufacturing} ($30$ lines $\times$ $40$ shifts, non-linear DGP, $9$ true edges), synthetic domains with their own driver/outcome catalogs and structural identities (\texttt{adjusted\_mortality}; \texttt{yield} $= 1 - $ \texttt{defect\_rate}). These test whether the framework and its conclusions transfer to differently-structured problems; we label them clearly as \emph{synthetic case studies}, not field deployments.
\end{itemize}

We deliberately do \emph{not} claim driver-ranking results on IHDP, LaLonde, Twins, or ACIC: those are single-treatment \emph{effect-estimation} benchmarks (one designated treatment with a known ATE), not structure-recovery tasks, so using them to score a ranked edge list would misrepresent both them and our method. Our IHDP loader is additionally a synthetic reconstruction. Extending MCES to effect-estimation baselines is left to future work.

\subsection{Evaluation Metrics}

Precision@$K$ and Recall@$K$ over the pooled driver--outcome grid; F1@$K$; Spearman's $\rho$ between CES and true edge strengths; expected calibration error (ECE) of CES against the empirical true-positive rate; and false-positive rate among null (non-causal) pairs. The primary scenario (\texttt{primary\_panel}) is averaged over $20$ seeds; the controlled decomposition scenario (E4) uses $3$ seeds and the \texttt{nonlinear} scenario a single seed, as noted per result. All averaged results are reported as mean\,$\pm$\,standard deviation.

\paragraph{Resampling policy and panel dependence.} The panel-aware methods respect the unit/time structure: mixed effects uses unit-level random intercepts, Granger and interrupted time series run per unit or per series, and transfer entropy permutes within units. The remaining associational and predictive methods (partial correlation, lasso, distance correlation, PPS, and RF+SHAP) pool observations across units and time and treat rows as exchangeable: their internal resampling (the partial-correlation $t$-test degrees of freedom, lasso and PPS $K$-fold cross-validation, and the distance-correlation permutation null) does \emph{not} adjust for within-unit or temporal dependence. Under such dependence the effective sample size is smaller than the nominal row count, so these methods' $p$-values and cross-validated scores are best read as \emph{screening statistics} rather than calibrated tests. Group-aware resampling (grouped/blocked cross-validation, within-unit permutation, cluster-robust inference) would tighten this and is left to future work.

\subsection{Experiments}

We run: \textbf{E1} ensemble vs.\ each individual method and leave-one-out ablation (\texttt{primary\_panel}); \textbf{E2} weight-scheme sensitivity (five schemes); \textbf{E3} a sample-size sweep over panel length and panel width around the primary setting; \textbf{E4} Structural--Behavioral Decomposition impact (controlled identity scenario); \textbf{E5} non-linear detection (\texttt{nonlinear}); calibration (E6), method-diversity (E7), and false-positive/FDR control with a CES-threshold sensitivity sweep (E8); \textbf{E9} out-of-sample predictive lift, a ground-truth-free consistency check; \textbf{E10} structure recovery on the \texttt{bnlearn} network benchmarks; plus the Sachs benchmark and the two additional domains.

\section{Results}\label{sec:results}

\subsection{E1: Does the Framework Identify the True Drivers?}\label{sec:results_e1}

The framework's job is to place genuinely causal driver--outcome pairs at the top of the CES ranking. It does: on \texttt{primary\_panel} the ensemble attains \textbf{Precision@5 $=$ \ensemblePfive} and \textbf{Precision@10 $=$ \ensemblePten} (mean over 20 seeds), the highest-scoring pairs are true edges, which is the outcome a practitioner cares about. This is the primary result.

A natural question is whether one could instead just pick a single method. Table~\ref{tab:e1} and Figure~\ref{fig:e1} report F1@10 for the ensemble and each method, and the answer is nuanced and worth stating plainly. We are also explicit about how this claim evolved: we initially expected the pooled score to \emph{outperform} the best individual method on accuracy, and it does not; the robustness framing reported below is what the experiments actually support, not what we set out to show. On this particular scenario a single method can match or edge out the pool on F1 (ensemble \ensembleFOne\ vs.\ best individual \bestIndivFOne; difference \ensembleVsBest). But this is not an argument for method selection, because \emph{which} method is best is not knowable in advance and changes across scenarios (Sections~\ref{sec:results_e5}, \ref{sec:results_domains}): the temporal methods that lead on lagged data are near-useless on this contemporaneous scenario, and vice versa. The ensemble's role is to be a method-agnostic default that avoids committing in advance to a single analytical tradition. We do not claim it is mathematically guaranteed to be optimal; we claim it removes the need for a prior method choice that, as the cross-scenario results show, is easy to get wrong. Leave-one-out ablation confirms the pool is genuinely distributed: removing any single method changes F1@10 by at most \maxLooDrop, the signature of a diversified estimator rather than one load-bearing method.

\begin{table}[ht]
\centering
\caption{E1, F1@10 on \texttt{primary\_panel} (mean\,$\pm$\,s.d.\ over 20 seeds). The ensemble is competitive with, but does not strictly dominate, the best individual method; its value is robustness across scenarios (Section~\ref{sec:results_e5}, \ref{sec:results_domains}).}
\label{tab:e1}
\small
\begin{tabular}{@{}lc@{}}
\toprule
\textbf{Method} & \textbf{F1@10} \\
\midrule
Ensemble (MCES) & \textbf{0.686} $\pm$ 0.035 \\
\midrule
Partial Corr. & 0.711 $\pm$ 0.016 \\
Lasso & 0.714 $\pm$ 0.000 \\
Distance Corr. & 0.475 $\pm$ 0.088 \\
PPS & 0.129 $\pm$ 0.029 \\
RF+SHAP & 0.539 $\pm$ 0.070 \\
Mixed Effects & 0.507 $\pm$ 0.166 \\
Granger & 0.107 $\pm$ 0.070 \\
ITS & 0.014 $\pm$ 0.029 \\
Transfer Entropy & 0.007 $\pm$ 0.021 \\
Bayesian Net & 0.321 $\pm$ 0.073 \\
Causal Forest & 0.521 $\pm$ 0.064 \\
\bottomrule
\end{tabular}

\end{table}

\begin{figure}[ht]
\centering
\includegraphics[width=0.62\textwidth]{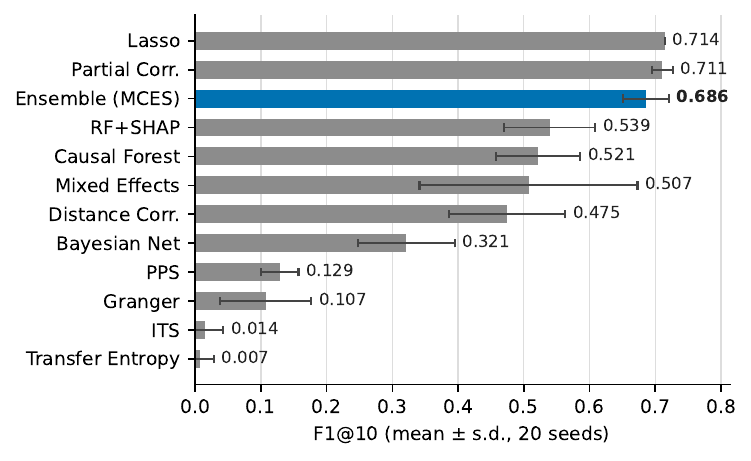}
\caption{E1, F1@10 of the ensemble (blue) against each individual method (gray) on \texttt{primary\_panel}. Error bars are s.d.\ over 20 seeds. The pool sits at the top of the range without depending on any single member.}
\label{fig:e1}
\end{figure}

\subsection{External Baselines: Structure Learning and Learned Weights}\label{sec:baselines}

The comparisons above are internal: the pool against its own members. Two external baselines test whether a different tool, or a smarter combination rule, dominates the pooled score (Table~\ref{tab:baselines}).

\paragraph{NOTEARS.} We fit linear NOTEARS \citep{zheng2018}, a continuous-optimization structure learner, on the pooled standardized panel of each scenario and rank the driver--outcome block of its weighted adjacency matrix by $|W_{jk}|$, granting it the same orientation restriction the MCES methods receive and scoring it with the same metrics. The pattern mirrors, at the level of an external state-of-the-art tool, exactly what Section~\ref{sec:results_e1} found internally: NOTEARS wins where its linearity assumption holds, edging the pool on the linear primary scenario ($0.70$ vs.\ \cesPrimaryFBl) and leading clearly on the (linear, hidden-confounder) \texttt{confounded} scenario ($0.82$ vs.\ $0.73$), and it fails hard where the assumption breaks, collapsing to \ntNonlinF\ against the pool's \cesNonlinF\ on \texttt{nonlinear} and trailing on healthcare ($0.80$ vs.\ $1.00$) and manufacturing ($0.63$ vs.\ $0.74$). On the real Sachs data both reach Precision@5 $= 1.0$ (NOTEARS is stronger deeper in the list, $0.80$ vs.\ $0.70$ at Precision@10). No tool dominates; the strongest external baseline we tried is, like the strongest internal member, scenario-dependent in exactly the way the method-agnostic-default argument predicts. NOTEARS also inherits the panel-dependence caveat of Section~\ref{sec:validation}: it treats pooled rows as exchangeable.

\paragraph{Learned weights.} We train a logistic-regression combiner on the eleven normalized method scores under the same rotating folds as the calibration experiment, a supervised upper reference that sees ground-truth labels the fixed-weight pool never uses. It attains Precision@10 $= 1.0$ on every held-out seed, for held-out F1@10 \stackedF\ against the fixed-weight \stackedCesF\ on identical seeds. The gain from supervision is thus about $+0.03$ F1, consistent with the weight-insensitivity of Section~\ref{sec:sensitivity}: learning the weights buys little beyond uniform pooling on this task, and doing so requires labeled causal ground truth that real deployments do not have.

\paragraph{Simple aggregation rules.} A skeptic may ask whether the weighted-pool formulation itself matters, or whether any agreement rule over the same eleven scores would do. We test three: a vote count of methods with nonzero gated evidence ($0.68$), the mean rank across methods ($0.69$), and the median normalized score ($0.70$), all within one standard deviation of CES ($0.69$) on the primary scenario (Table~\ref{tab:baselines}). We report this plainly: on this scenario, the specific synthesis formula is not the source of the performance, the multi-lens agreement is, and the weighted pool's advantages are operational rather than accuracy-based (interpretable weights, applicability-aware renormalization, and the convergence-band semantics of Section~\ref{sec:convclass}).

\begin{table}[ht]
\centering
\caption{External baselines, F1@10. The baseline in the scenario rows is NOTEARS-linear \citep{zheng2018}, fit per scenario on the pooled panel; in the final row it is a supervised logistic combination of the eleven method scores (rotating folds; superscript a: trained on ground-truth labels, an upper reference available only when labels exist).}
\label{tab:baselines}
\small
\begin{tabular}{@{}lcc@{}}
\toprule
\textbf{Scenario} & \textbf{MCES F1@10} & \textbf{Baseline F1@10} \\
\midrule
Primary panel (20 seeds) & 0.69 $\pm$ 0.04 & 0.70 $\pm$ 0.04 \\
Non-linear & 0.74 $\pm$ 0.00 & 0.42 $\pm$ 0.00 \\
Confounded & 0.73 $\pm$ 0.00 & 0.82 $\pm$ 0.00 \\
Healthcare & 1.00 $\pm$ 0.00 & 0.80 $\pm$ 0.00 \\
Manufacturing & 0.74 $\pm$ 0.00 & 0.63 $\pm$ 0.00 \\
\midrule
Stacked logistic (rotating folds) & 0.69 $\pm$ 0.04 & 0.71 $\pm$ 0.00\textsuperscript{a} \\
Vote count (methods $> 0$) & 0.69 $\pm$ 0.04 & 0.68 $\pm$ 0.04 \\
Mean rank across methods & 0.69 $\pm$ 0.04 & 0.69 $\pm$ 0.04 \\
Median normalized score & 0.69 $\pm$ 0.04 & 0.70 $\pm$ 0.03 \\
\bottomrule
\end{tabular}

\end{table}

\subsection{E2: Robustness to Weights}\label{sec:sensitivity}

Table~\ref{tab:e2} varies the weighting across five schemes (uniform, tiered, tier-heavy, associational-heavy, causal-only). The mean pairwise Spearman correlation between CES rankings is \weightSpearman\ (minimum \weightSpearmanMin): the rankings remain broadly similar across schemes, although F1@10 does change under the extreme schemes and declines under the causal-only weighting. The best-performing scheme is in fact the \emph{associational-heavy} one, and causal-only is worst; this is informative about what the score is doing on this scenario, where effects are contemporaneous: CES's strength here comes primarily from robust relevance detection across lenses, not from privileged causal identification, consistent with the interpretation of Section~\ref{sec:norm}. The tiered scheme provides no measurable accuracy advantage over uniform weighting. Its only distinct property is that, when all eleven methods are active, Moderate-or-higher CES requires Tier~3 or Tier~4 evidence; because renormalization can remove that guarantee, uniform weighting remains the default and the Strong-convergence directional requirement is enforced explicitly (Section~\ref{sec:convclass}). We state the invariance plainly rather than present it as a benefit of tiering.

\begin{table}[ht]
\centering
\caption{E2, F1@10 by weight scheme on \texttt{primary\_panel}, with mean pairwise rank correlation between schemes.}
\label{tab:e2}
\small
\begin{tabular}{@{}lc@{}}
\toprule
\textbf{Weight scheme} & \textbf{F1@10} \\
\midrule
uniform & 0.686 $\pm$ 0.035 \\
tiered & 0.646 $\pm$ 0.042 \\
tier heavy & 0.629 $\pm$ 0.058 \\
associational heavy & 0.707 $\pm$ 0.021 \\
causal only & 0.564 $\pm$ 0.081 \\
\midrule
\multicolumn{2}{@{}l}{Mean pairwise rank corr. $\rho = 0.945$ (min 0.814)} \\
\bottomrule
\end{tabular}

\end{table}

\subsection{E3: Sample-Size Sensitivity}\label{sec:results_e3}

Several limitations stated in this paper attribute method weakness to sample size; E3 measures that dependence directly rather than asserting it. We sweep panel length ($T \in \{5, 10, 20, 40\}$ at $N=23$) and panel width ($N \in \{6, 12, 23, 46\}$ at $T=20$) on the primary scenario, three seeds per cell (Table~\ref{tab:e3}). The results corrected our own expectation: we anticipated that short time series would be the binding constraint, and they are not, on this scenario. All eleven methods remain nominally active in every cell (the applicability gates do not trigger), so degradation is a matter of statistical power, not method dropout. Panel length costs comparatively little: F1@10 falls only from \eThreeLongTF\ at $T=40$ to \eThreeShortTF\ at $T=5$, and the metric saturates by $T=20$. Panel \emph{width} is the binding constraint: at $N=6$ the ensemble attains only F1@10 $=$ \eThreeNarrowNF\ (with high seed-to-seed variance), while at $N=46$ it reaches \eThreeWideNF\ with Precision@10 $= 1.00$ across all seeds. This is consistent with the scenario's structure: the true effects are contemporaneous and linear with unit-specific intercepts, so cross-sectional replication is the scarce resource, and adding periods beyond $T \approx 20$ adds little. On temporally-lagged data the roles would plausibly reverse; we measured this scenario, and the claim is scoped to it.

\begin{table}[ht]
\centering
\caption{E3, sample-size sweep on \texttt{primary\_panel} (mean\,$\pm$\,s.d.\ over 3 seeds). Top block: panel length $T$ at $N=23$; bottom block: panel width $N$ at $T=20$. ``Methods'' is the mean number of methods that ran; all eleven remain active in every cell, so the degradation is statistical power, not applicability.}
\label{tab:e3}
\small
\begin{tabular}{@{}rrrrcc@{}}
\toprule
$N$ & $T$ & Obs. & Methods & \textbf{P@10} & \textbf{F1@10} \\
\midrule
23 & 5 & 115 & 11.0 & 0.83 $\pm$ 0.09 & 0.60 $\pm$ 0.07 \\
23 & 10 & 230 & 11.0 & 0.87 $\pm$ 0.05 & 0.62 $\pm$ 0.03 \\
23 & 20 & 460 & 11.0 & 0.97 $\pm$ 0.05 & 0.69 $\pm$ 0.03 \\
23 & 40 & 920 & 11.0 & 0.97 $\pm$ 0.05 & 0.69 $\pm$ 0.03 \\
\midrule
6 & 20 & 120 & 11.0 & 0.47 $\pm$ 0.17 & 0.33 $\pm$ 0.12 \\
12 & 20 & 240 & 11.0 & 0.83 $\pm$ 0.12 & 0.60 $\pm$ 0.09 \\
23 & 20 & 460 & 11.0 & 0.97 $\pm$ 0.05 & 0.69 $\pm$ 0.03 \\
46 & 20 & 920 & 11.0 & 1.00 $\pm$ 0.00 & 0.71 $\pm$ 0.00 \\
\bottomrule
\end{tabular}

\end{table}

\subsection{E4: Structural--Behavioral Decomposition}\label{sec:exp_decomp}

We construct the exact setting Section~\ref{sec:decomp} identifies as the danger zone: identity components \texttt{price} and \texttt{volume} are candidate drivers, the outcome \texttt{revenue}$=$\texttt{price}$\times$\texttt{volume} is a pure algebraic identity with no behavioral cause, and all genuine causal structure lives in a separate outcome. Table~\ref{tab:e4} shows the effect. Without decomposition, the two identity pairs occupy the top of the ranking and top-of-list precision (Precision@3) is only \decompPthreeWithout; removing the identity edges raises it to \decompPthreeWith, with zero identity pairs surviving in the top~5. We state the scope of this gain precisely: Precision@5 is unchanged ($0.600$ with and without), so the effect is specifically the removal of the two tautological pairs from the top ranks, not a general accuracy improvement, which is exactly what the mechanism predicts. This isolates the contribution: decomposition matters precisely, and only, when components and their algebraic composite are both candidates, consistent with the scope stated in Section~\ref{sec:decomp}, and a no-op otherwise (as it is on \texttt{primary\_panel}, where identity components are not drivers).

\begin{table}[ht]
\centering
\caption{E4, Structural--Behavioral Decomposition on a controlled identity scenario (mean over 3 seeds).}
\label{tab:e4}
\small
\begin{tabular}{@{}lcc@{}}
\toprule
 & \textbf{Without SBD} & \textbf{With SBD} \\
\midrule
Precision@3 & 0.333 & 1.000 \\
Precision@5 & 0.600 & 0.600 \\
Identity pairs in top 5 & 2.0 & 0.0 \\
\bottomrule
\end{tabular}

\end{table}

\subsection{E5: Non-Linear Detection and Why Diversity Helps}\label{sec:results_e5}

The \texttt{nonlinear} scenario is a \emph{single-run diagnostic} (one seed), reported without a standard deviation. On it (Table~\ref{tab:e5}), \emph{monotone} non-linear edges ($\sqrt{}$, log) are recovered by most methods, including partial correlation, while the \emph{symmetric} forms (quadratic, threshold) are missed broadly, so the per-method recalls cluster and the ensemble's recall (\nlEnsembleRecall) \emph{matches} the strongest individual methods rather than exceeding them. Transfer entropy and ITS contribute least here (near-zero recall), consistent with their sensitivity to sample size and to the absence of a declared break. The honest reading reinforces the robustness thesis without overstating it: the pool attains the best available recall \emph{without the analyst having to know in advance which method that is}, but it does not manufacture detection power that no member possesses. Distance correlation is the member designed for symmetric dependence, and its inclusion is the pool's main defense here; where even it lacks power at this sample size, symmetric non-linearities remain hard for the entire pool.

\begin{table}[ht]
\centering
\caption{E5, recall of non-linear edges (top-$n_{\text{true}}$) on the \texttt{nonlinear} scenario.}
\label{tab:e5}
\small
\begin{tabular}{@{}lc@{}}
\toprule
\textbf{Method} & \textbf{Recall of non-linear edges} \\
\midrule
Ensemble (MCES) & \textbf{0.667} \\
\midrule
Partial Corr. & 0.667 \\
Lasso & 0.667 \\
Distance Corr. & 0.667 \\
PPS & 0.444 \\
RF+SHAP & 0.667 \\
Mixed Effects & 0.444 \\
Granger & 0.556 \\
ITS & 0.111 \\
Transfer Entropy & 0.000 \\
Bayesian Net & 0.556 \\
Causal Forest & 0.444 \\
\bottomrule
\end{tabular}

\end{table}

\subsection{Calibration}\label{sec:calibration}

Raw CES is an ordinal, within-outcome score (Section~\ref{sec:norm}); we ask how far it is from a probability, and evaluate this \emph{on held-out data with rotating splits} rather than a single hand-chosen partition. We fit an isotonic regression $\hat{c}: \CES \mapsto \widehat{\Prob}(\text{true edge})$ under \calFolds-fold cross-validation over the twenty \texttt{primary\_panel} seeds: each fold trains on the other sixteen seeds plus the \texttt{nonlinear} and \texttt{confounded} scenarios and evaluates on its four held-out seeds (pooled base rate of true edges \calBaseRate). We lead with the Brier score, because at a base rate this low the expected calibration error is a weak metric (a calibrator that outputs the base rate everywhere scores well on ECE by construction). Across folds, calibration improves the held-out Brier score from \calBrierRaw\ to \calBrierCal; held-out ECE moves from \calECE\ to \calECEcal. Figure~\ref{fig:calibration} shows the reliability curves over the pooled held-out predictions of all folds. We stress that this is a \emph{within-distribution} result: the map is scenario-specific, and its transportability to a genuinely new target domain is not established and would require separate evaluation there.

\begin{figure}[ht]
\centering
\includegraphics[width=0.42\textwidth]{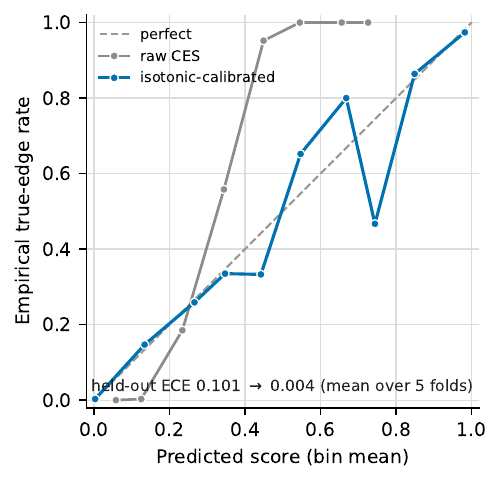}
\caption{Held-out reliability diagram: raw CES (gray) and isotonic-calibrated CES (blue) against the empirical true-edge rate, over the pooled held-out predictions of all rotating folds (no fold's calibrator sees its own test seeds). The dashed line is perfect calibration; the calibrated curve lies closer to it.}
\label{fig:calibration}
\end{figure}

\subsection{Method Diversity}\label{sec:diversity}

The variance argument of Section~\ref{sec:variance} predicts a benefit only if the methods are actually diverse. As a descriptive indication, on \texttt{primary\_panel} the mean pairwise correlation of per-pair normalized scores is \meanMethodCorr\ (max \maxMethodCorr), which shows the methods rank candidate pairs quite differently. The variance proposition, however, concerns correlation across \emph{repeated draws of the data-generating process for a fixed pair}, a different quantity, and the twenty independent seeds of the primary scenario are exactly such repeated draws. Measuring the proposition's own quantity, for each of the $570$ candidate pairs we correlate every method pair's scores across the twenty seeds and average: the fixed-pair mean cross-method correlation is $\bar\rho = 0.13$, even lower than the cross-pair diagnostic. The low-correlation premise of Proposition~\ref{prop:main} therefore holds on this scenario by direct measurement, not assumption; what remains unestablished is the further step from reduced score variance to improved ranking accuracy, which the E1 results show is not automatic. Notably, the Granger--transfer-entropy correlation is only \grangerTeCorr\ here: although \citet{barnett2009} show the two coincide for Gaussian data, on this predominantly contemporaneous scenario both temporal methods are near-inactive, so their empirical scores do not co-move. The redundancy that theory warns about is thus data-dependent (it would appear on temporally-rich, near-Gaussian data), which is precisely why we recommend measuring $\bar\rho$ per dataset rather than assuming a fixed redundancy structure.

\subsection{Empirical False-Positive Behavior and Threshold Sensitivity}\label{sec:exp_fp}

We report false-positive control at two levels and make Benjamini--Hochberg (BH) gating the primary per-method setting rather than an option. Among null (non-causal) pairs on \texttt{primary\_panel}, the mean per-method false-positive rate (nonzero evidence on a null pair) is \fpRaw\ unadjusted; applying BH gating across the driver--outcome grid reduces it to \fpBH, and we recommend BH as the default because the unadjusted per-method gates do not account for the size of the grid. The more consequential quantity is at the ensemble level: the fraction of null pairs reaching Moderate-or-higher convergence ($\CES \ge 0.4$) is \fpEnsemble, because requiring cross-tradition agreement is itself a stringent filter, a pair that passes one method's gate by chance rarely passes several.

Because the $0.4$ cutoff is an interpretive design choice (Section~\ref{sec:convclass}), a fair concern is that this headline rate might hold only at the chosen threshold. Table~\ref{tab:e8thresh} therefore sweeps the threshold and reports, at each candidate value, both the null-pair rate (false positives) and the true-pair rate (retention). The threshold constants were fixed before this sweep was run; the sweep is the audit, not the selection procedure. The false-positive property is not an artifact of where the Moderate band sits: even at the loosest threshold examined ($\CES \ge 0.3$) the null-pair rate is only \fpEnsembleThree. Nor is it specific to the primary panel: the same computation on the \texttt{nonlinear}, \texttt{confounded}, healthcare, and manufacturing scenarios yields a null-pair rate of $0.000$ at the Moderate threshold in every case, so ``on the evaluated scenarios'' is a measured statement. The retention column shows the cost side of the same tradeoff, and we state it plainly: only \tpEnsembleFour\ of true pairs reach the default Moderate threshold (\tpEnsembleThree\ at $0.3$), and the Strong band ($\CES \ge 0.7$) is reached by almost no pairs, true or null, on this scenario. The bands are conservative by construction: crossing them is strong evidence, but failing to cross them is weak evidence of absence. We describe this as strong empirical false-positive \emph{control} on the evaluated scenarios; we do not prove a formal false-discovery bound, and do not claim one.

\paragraph{All-null negative control.} Relative-max normalization is most vulnerable when an outcome has \emph{no} true driver, since some pair is always the per-outcome maximum; we construct that worst case directly. Three panels identical in shape to the primary scenario (same driver generator, $N=23$, $T=20$) have six outcomes of pure unit-level AR(1) noise, so all $570$ candidate pairs are null. The vulnerability is real and bounded: \nullModRate\ of null pairs reach Moderate, roughly sixteen times the rate on the primary panel (where genuine drivers occupy the maximum slots) yet still below half a percent, and the top-ranked pair of a fully null outcome averages CES \nullMaxCes, sitting at the Moderate boundary; the Strong band is never reached. The practical guidance follows directly: a single Moderate pair atop an otherwise quiet outcome is exactly the pattern this control produces from noise, and should be treated as a prompt for the E9 lift diagnostic rather than as a finding.

\begin{table}[ht]
\centering
\caption{E8, CES-threshold sensitivity on \texttt{primary\_panel} (mean\,$\pm$\,s.d.\ over 20 seeds): fraction of null pairs and of true pairs at or above each candidate threshold.}
\label{tab:e8thresh}
\small
\begin{tabular}{@{}lcc@{}}
\toprule
\textbf{CES threshold} & \textbf{Null pairs $\ge$ t} & \textbf{True pairs $\ge$ t} \\
\midrule
$\ge 0.3$ & 0.0051 $\pm$ 0.0033 & 0.69 $\pm$ 0.09 \\
$\ge 0.35$ & 0.0012 $\pm$ 0.0012 & 0.60 $\pm$ 0.06 \\
$\ge 0.4$ (default) & 0.0003 $\pm$ 0.0006 & 0.50 $\pm$ 0.05 \\
$\ge 0.45$ & 0.0002 $\pm$ 0.0005 & 0.41 $\pm$ 0.06 \\
$\ge 0.5$ & 0.0000 $\pm$ 0.0000 & 0.33 $\pm$ 0.07 \\
$\ge 0.6$ & 0.0000 $\pm$ 0.0000 & 0.19 $\pm$ 0.04 \\
$\ge 0.7$ & 0.0000 $\pm$ 0.0000 & 0.01 $\pm$ 0.03 \\
\bottomrule
\end{tabular}

\end{table}

\subsection{E9: Out-of-Sample Predictive Lift}\label{sec:exp_lift}

Validation against ground truth is impossible on real deployments, so we add a consistency check that requires none: if CES tracks genuine causal relevance, then high-CES drivers should carry \emph{incremental out-of-sample predictive value} for their outcome beyond the outcome's own history, and null pairs should not. For every behavioral (driver, outcome) pair we fit a pooled autoregression $y_t \sim y_{t-1}$ (units demeaned with training-period means) and measure the change in out-of-sample $R^2$, on a chronological hold-out, from adding $(x_t, x_{t-1})$. The design is \emph{nested}: the CES used here is computed from method runs on only the chronological training window (the first ${\sim}70\%$ of periods), so no component of the score sees the held-out periods on which lift is evaluated. On \texttt{primary\_panel}, the group-level contrast is clear: pairs with $\CES \ge 0.4$ show a mean lift of \liftHigh, while null pairs with $\CES < 0.1$ show \liftNull\ (Figure~\ref{fig:e9}). The pairwise signal is much weaker: the rank correlation between CES and lift across all pairs is only \liftSpearman, so CES separates the high-evidence group from the null group but does not finely order pairs by their predictive lift. The check is deliberately one-directional: predictive lift does not certify causation (a strong confounder also predicts), but its \emph{absence} for a high-CES pair is a red flag. Because it needs no ground truth, this diagnostic can be computed on real data, subject to the same temporal-split and data-quality assumptions, and we recommend it as a standard companion to CES.

\begin{figure}[ht]
\centering
\includegraphics[width=0.62\textwidth]{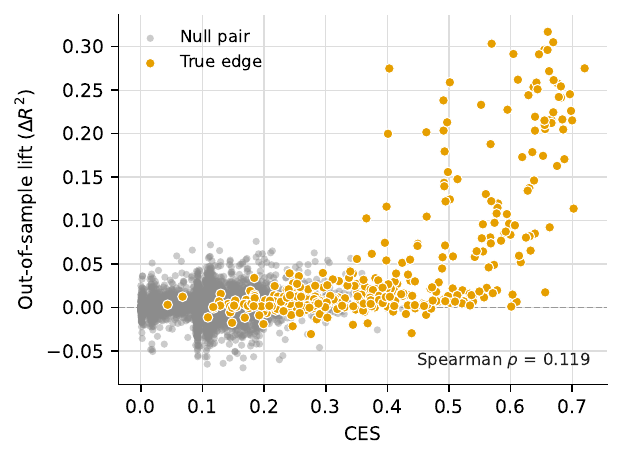}
\caption{E9, out-of-sample predictive lift vs.\ CES on \texttt{primary\_panel} (all behavioral pairs, 20 seeds). High-CES pairs concentrate at positive lift; null pairs at zero.}
\label{fig:e9}
\end{figure}

\section{Domain Applications}\label{sec:applications}\label{sec:results_domains}

We report the primary scenario and two additional synthetic domains. We describe these as \emph{case studies on synthetic data with known ground truth}, not field results; the value is in showing the framework and its robustness conclusion transfer across problem structures.

\subsection{Application A: Multi-Site System Rollout (Primary Panel)}

The \texttt{primary\_panel} scenario is a synthetic multi-site operations panel: $23$ units observed over $20$ periods during a phased rollout of a new operational system, with six outcome KPIs and $95$ numeric candidate drivers (of $99$ cataloged; four are categorical site attributes) spanning system performance, feature adoption, training, workflow, support, and staffing. A structural identity of the form Revenue $=$ Average Transaction Value $\times$ Units $\times$ Days is declared; because its components are outcomes rather than candidate drivers here, decomposition removes them from outcome--outcome consideration but does not alter the driver--outcome ranking (Section~\ref{sec:exp_decomp}). Results are those of E1--E3 above: the ensemble ranks the true behavioral drivers at the top (Precision@10 in the high range) while remaining robust to weighting.

\subsection{Applications B \& C: Healthcare and Manufacturing (Synthetic)}

Table~\ref{tab:domains_results} summarizes both. They reinforce the paper's central, deliberately unflattering point: the ensemble is not uniformly the single best method. In \textbf{healthcare} it \emph{matches} the best individual method (both F1@10 \hcF); in \textbf{manufacturing}, where several true edges are symmetric non-linearities, a single method is clearly ahead (\mfgF\ vs.\ \mfgBest). The ensemble is never the worst and never collapses, but it does not strictly dominate in either domain. The value proposition is therefore explicitly \emph{not} ``the ensemble wins''; it is that the identity of the winning method changes across healthcare, manufacturing, and the primary scenario, so the pool avoids requiring the analyst to commit to a fixed method before the data-generating regime is known. RF+SHAP is consistently strong; temporal methods matter only where lags or breaks exist.

\begin{table}[ht]
\centering
\caption{Domain applications (synthetic, full ensemble). F1@10 with the best individual method for comparison.}
\label{tab:domains_results}
\small
\begin{tabular}{@{}lccc@{}}
\toprule
\textbf{Domain} & \textbf{\# true edges} & \textbf{Ensemble F1@10} & \textbf{Best individual F1@10} \\
\midrule
Healthcare (synthetic) & \hcEdges & \hcF & \hcBest \\
Manufacturing (synthetic) & \mfgEdges & \mfgF & \mfgBest \\
\bottomrule
\end{tabular}
\end{table}

\subsection{Sachs Real-Data Benchmark}\label{sec:results_sachs}

On the \emph{real} Sachs flow-cytometry data (seven applicable methods, cross-sectional), MCES attains Precision@5 $=$ \sachsPfive\ against the consensus network, with its top-ranked pairs (in CES order) being canonical signaling edges: PKC$\to$P38, Erk$\to$Akt, PKC$\to$Jnk, PIP3$\to$PIP2, and Raf$\to$Mek. Precision falls off deeper in the list (Precision@10 $=0.70$), as expected: the temporal and panel methods cannot run on cross-sectional data, so only seven of the eleven methods contribute, and the biological network is dense with feedback that observational scoring cannot orient. Still, obtaining top-5 precision of $1.0$ on data generated by a process entirely outside our control is meaningful external corroboration.

\subsection{Bayesian-Network Structure Benchmarks}\label{sec:results_bnlearn}

To test recovery beyond a single network, we run the seven cross-sectional methods on \bnNetCount\ standard \texttt{bnlearn} benchmarks spanning $20$ to $76$ nodes across six unrelated domains, each with an exact published DAG (Table~\ref{tab:bnlearn}). Three observations hold across the suite. First, top-of-list precision is strong \emph{given the declared partition} (which on these benchmarks is derived from the gold DAG; the audit below removes it): Precision@5 is $1.0$ on five of the six networks (including the $70$-node \textsc{Hepar2} and $76$-node \textsc{Win95pts}), dropping only on the densely-connected $56$-node \textsc{Hailfinder} ($0.6$), so on most networks the highest-CES pairs are genuine edges regardless of scale, the property a practitioner relies on. Second, full-edge-set recovery (F1@K, $0.40$ to $0.75$) does not track raw node count: the $76$-node \textsc{Win95pts} is recovered better than the $56$-node \textsc{Hailfinder}, and the smallest network (\textsc{Child}) is recovered best. We do not have a structural account of \textsc{Hailfinder}'s difficulty, and we note that simple graph statistics do not supply one: \textsc{Win95pts} has a \emph{higher} average degree ($1.47$ vs.\ $1.18$) and a higher mean parent count among non-root nodes ($2.67$ vs.\ $1.69$) yet is recovered better, so neither density nor parent sharing explains the gap. Distributional properties of the sampled data (many-state variables, skewed conditional distributions at $1{,}000$ samples) are a plausible cause we have not isolated. Third, and consistent with the rest of the paper, the ensemble neither wins big nor collapses: it leads the best individual method on \textsc{Hepar2}, ties it on \textsc{Child}, and trails it by at most about $0.11$ F1 elsewhere, remaining a method-agnostic default without the analyst having to know which method will lead on a given network or domain. These are sampled-data benchmarks, but their \emph{structure} is external, published, and was not designed with our methods in mind, so together with Sachs they broaden the benchmark evaluation across externally defined structures and domains.

\paragraph{Orientation audit: the declared partition is load-bearing.} The results above are for the task MCES actually targets, ranking a \emph{declared} driver--outcome grid, but on these benchmarks the declaration itself was derived from the published DAG (each node's role is assigned by its net edge direction, the same convention as Sachs). That is a form of ground-truth assistance a practitioner doing general discovery would not have, so we audit it: Table~\ref{tab:bnlearnU} repeats the evaluation with \emph{no} role assignment, every node is both candidate driver and candidate outcome, all ordered pairs compete, and a reverse-oriented edge counts as an error. Performance drops sharply: Precision@5 falls to $0.2$--$0.6$ and ensemble F1@K to $0.25$--$0.35$, with the best individual method similarly reduced. The conclusion is one we state plainly rather than bury: \textbf{the declared partition carries substantial information, and MCES does not recover edge orientation on its own}, as most of its members are direction-symmetric. In the intended use case the partition is genuine domain knowledge (analysts know which variables are interventions and which are KPIs), and the restricted results above measure exactly that setting; the unrestricted results measure general discovery, a task MCES does not claim, and on which specialized structure learners with orientation rules are the appropriate tools. Reverse-direction penalization (Section~\ref{sec:futurework}) is the natural extension.

\begin{table}[ht]
\centering
\caption{E10b, orientation audit: the same networks with no gold-DAG role assignment (all ordered pairs; reverse orientations are errors). Compare Table~\ref{tab:bnlearn}: removing the declared partition sharply reduces performance for the ensemble and the best individual method alike.}
\label{tab:bnlearnU}
\small
\begin{tabular}{@{}lrrccc@{}}
\toprule
\textbf{Network} & \textbf{Nodes} & \textbf{Pairs} & \textbf{P@5} & \textbf{Ens.\ F1@K} & \textbf{Best ind.\ F1@K} \\
\midrule
Child & 20 & 380 & 0.20 & 0.32 & 0.40 \\
Insurance & 27 & 702 & 0.20 & 0.25 & 0.33 \\
Alarm & 37 & 1332 & 0.20 & 0.35 & 0.35 \\
Hailfinder & 56 & 3080 & 0.60 & 0.30 & 0.35 \\
Hepar2 & 70 & 4830 & 0.40 & 0.28 & 0.31 \\
Win95pts & 76 & 5700 & 0.60 & 0.34 & 0.31 \\
\bottomrule
\end{tabular}

\end{table}

\begin{table}[ht]
\centering
\caption{Bayesian-network structure benchmarks (\texttt{bnlearn}). $K$ is the number of scoreable forward edges under the driver/outcome orientation. The ensemble tracks the best individual method at each scale.}
\label{tab:bnlearn}
\small
\begin{tabular}{@{}lccccc@{}}
\toprule
\textbf{Network} & \textbf{Nodes} & \textbf{Scoreable} & \textbf{P@5} & \textbf{Ens.\ F1@K} & \textbf{Best ind.\ F1@K} \\
\midrule
Child & 20 & 12/25 & 1.00 & 0.75 & 0.75 \\
Insurance & 27 & 28/52 & 1.00 & 0.68 & 0.79 \\
Alarm & 37 & 28/46 & 1.00 & 0.71 & 0.79 \\
Hailfinder & 56 & 40/66 & 0.60 & 0.40 & 0.47 \\
Hepar2 & 70 & 95/123 & 1.00 & 0.49 & 0.47 \\
Win95pts & 76 & 82/112 & 1.00 & 0.63 & 0.65 \\
\bottomrule
\end{tabular}

\end{table}

\section{Discussion}\label{sec:discussion}

\subsection{Contributions}

\begin{enumerate}[nosep]
  \item A \textbf{quantitative operationalization of causal triangulation} from raw observational data, turning a qualitative recommendation into a computable convergence score that pools methods across distinct mathematical traditions.
  \item \textbf{Structural--Behavioral Decomposition}, with a precise statement of \emph{when} it matters (identity components must be candidate drivers) and a controlled demonstration that it then raises top-of-list precision from \decompPthreeWithout\ to \decompPthreeWith.
  \item \textbf{A characterization of when synthesis helps.} Rather than claim uniform dominance, we show the ensemble's value is that it is a method-agnostic default: no single method is best across scenarios, so pooling avoids an easily-wrong prior method choice, backed by a variance argument (Section~\ref{sec:variance}) and a diversity measurement (Section~\ref{sec:diversity}). We frame this as avoided commitment rather than proven optimality.
  \item \textbf{Scenario-specific calibration and multiple-testing controls}: an isotonic calibration of CES to an empirical true-edge rate on held-out synthetic data (whose transportability to new domains we do not claim), and Benjamini--Hochberg gating across the driver--outcome grid that reduces the per-method false-positive rate.
  \item A \textbf{reference implementation} with three synthetic domains and a real-data benchmark loader (available from the authors on request).
\end{enumerate}

\subsection{Limitations}

We state these plainly; several qualify claims made above.
\begin{enumerate}[nosep]
  \item \textbf{The ensemble is not uniformly the best method.} On some scenarios a single method, or an external structure learner (Section~\ref{sec:baselines}), edges it out on F1; the case for MCES is robustness across scenarios and false-positive control, not universal dominance. This framing is itself a revision: accuracy improvement was the original design goal, and the evidence did not support it.
  \item \textbf{Weighting is near-irrelevant to accuracy.} We default to uniform weights; the tiered alternative yields no measurable F1 gain (Section~\ref{sec:sensitivity}) and is offered only for its interpretability property.
  \item \textbf{Within-outcome relative normalization.} Raw CES is ordinal and comparable primarily within an outcome. Cross-outcome probability interpretation requires calibration and validation for the target distribution (Section~\ref{sec:calibration}).
  \item \textbf{Shared data, correlated errors.} All methods see the same data; the variance benefit depends on measured diversity ($\bar\rho \approx$ \meanMethodCorr\ here) and would shrink on data where methods become redundant (e.g.\ Granger and transfer entropy on Gaussian temporal series).
  \item \textbf{Hidden confounding, feedback, colliders} bias every observational method and are not removed by pooling (Section~\ref{sec:failures}). MCES prioritizes hypotheses; it does not replace experiments.
  \item \textbf{Decomposition requires domain knowledge} to enumerate identities, and is a no-op when components are not candidate drivers.
  \item \textbf{The declared driver/outcome partition is load-bearing.} MCES ranks a declared grid; it does not recover edge orientation. The E10b audit (Section~\ref{sec:results_bnlearn}) shows that removing the partition on the network benchmarks drops Precision@5 from $1.0$ to $0.2$--$0.6$: where the partition is not genuine domain knowledge, MCES's headline precision does not apply.
  \item \textbf{Applicability limits}: cross-sectional data disables the temporal and panel methods (as on Sachs, leaving seven); small $N$/$T$ reduces power (quantified at the ensemble level by the E3 sweep, Section~\ref{sec:results_e3}, where panel width was the binding constraint on the primary scenario); ITS needs a break point.
  \item \textbf{Panel dependence in the non-panel methods}: partial correlation, lasso, distance correlation, PPS, and RF+SHAP pool observations across units and time, so their inferential outputs are not adjusted for within-unit or temporal dependence (Section~\ref{sec:validation}) and should be read as screening statistics; group-aware resampling is future work.
  \item \textbf{Heuristic convergence thresholds}: the $0.4$ and $0.7$ band cutoffs are interpretive design choices. Section~\ref{sec:exp_fp} now audits them with a threshold sweep on the evaluated scenarios (the false-positive property is stable across candidate cutoffs), but the cutoffs remain conventions, not estimated quantities, and their interpretation does not transfer to new domains without the same audit.
  \item \textbf{Multiple testing across the grid} is controlled by Benjamini--Hochberg gating by default (Section~\ref{sec:exp_fp}); we report unadjusted per-method rates only as a comparison. We control the empirical false-positive rate on the evaluated scenarios but prove no formal false-discovery bound.
  \item \textbf{Computational cost}: eleven methods over $M \times K$ pairs is expensive (minutes on our scenarios; the causal forest and transfer entropy dominate).
\end{enumerate}

\subsection{When NOT to Use MCES}

\begin{itemize}[nosep]
  \item When you \emph{can} run a controlled experiment, do that instead.
  \item When $N < 10$ units, panel methods lose power; at $N = 1$ MCES runs in the degraded single-unit mode (Section~\ref{sec:weights}) with correspondingly narrower evidence coverage and fewer applicable methods.
  \item When $T < 5$ periods, insufficient for time-series methods.
  \item When the relationship is obviously structural, no statistics needed.
  \item When only cross-sectional data exists (no panel structure), Granger, Transfer Entropy, and ITS are inapplicable.
\end{itemize}

\subsection{Comparison to Alternative Approaches}

\begin{table}[ht]
\centering
\caption{When to use alternative approaches instead of MCES.}
\label{tab:alternatives}
\small
\begin{tabularx}{\textwidth}{@{}lY@{}}
\toprule
\textbf{Approach} & \textbf{When to Use Instead of MCES} \\
\midrule
A/B test & When randomization is possible \\
Single well-specified causal model & Strong theory + large $N$ + one specific question \\
CausalTune / method selection & Computational budget is limited, need one fast answer \\
Causal-Copilot & Need guided exploration, not rigorous synthesis \\
\bottomrule
\end{tabularx}
\end{table}

\subsection{Future Work}\label{sec:futurework}

\begin{enumerate}[nosep]
  \item \textbf{Learned weights and redundancy pruning:} learn weights from labeled causal datasets and jointly down-weight redundant pairs (e.g.\ Granger/transfer entropy on near-Gaussian data), which our diversity analysis flags.
  \item \textbf{Effect-estimation benchmarks:} extend the evaluation to official semi-synthetic benchmark datasets such as IHDP and ACIC, using treatment-effect metrics appropriate to their single-treatment setting (our current loaders for these are synthetic and are excluded from claims).
  \item \textbf{Reverse-direction penalization:} run the directional methods in both orientations and penalize pairs with stronger reverse evidence, relaxing the a-priori driver/outcome split.
  \item \textbf{Method expansion:} propensity-score matching, instrumental variables, and difference-in-differences via DoWhy.
  \item \textbf{LLM integration:} use LLMs to propose the structural decomposition and to narrate CES results, following Causal-Copilot \citep{causalcopilot2025}.
\end{enumerate}

\section{Conclusion}\label{sec:conclusion}

Observational causal inference is among the most ubiquitous challenges across science and industry. Practitioners overwhelmingly rely on single analytical methods, each with well-documented blind spots. Recent advances in automated causal inference have focused on \emph{selecting} the optimal method for a dataset or \emph{aggregating} instances of the same algorithm, but neither addresses the fundamental limitation of method-specific failure modes.

MCES operationalizes the qualitative recommendation to ``use multiple methods and see if they agree'' \citep{munafio2018}. By running eleven methods across eight distinct mathematical traditions and pooling their non-commensurable outputs into a Convergent Evidence Score, it produces a reproducible ranking of candidate driver--outcome relationships by cross-method agreement. Structural--Behavioral Decomposition prevents a specific class of false positives, algebraic-identity edges, exactly in the setting where identity components are candidate drivers.

Empirically, on synthetic ground truth, the real Sachs benchmark, six external Bayesian-network structure benchmarks, and two additional synthetic domains, MCES places true edges at the top of its ranking (Precision@5 $=$ \ensemblePfive\ on the primary scenario) while showing a low empirical rate of null pairs reaching Moderate-or-higher convergence on the evaluated scenarios. The evaluation is deliberately even-handed: the pool does not uniformly dominate individual methods (on some scenarios a single method scores higher), and its value is that it is a \emph{method-agnostic default}, since no individual method is uniformly best across the evaluated scenarios, so pooling avoids committing in advance to a single analytical tradition when the right choice is unknown.

The central takeaway is modest by design. MCES summarizes whether heterogeneous analytical approaches converge on the same candidate relationship. It is intended to prioritize hypotheses when method choice is uncertain, not to replace experimental identification or a well-specified causal design, and its evidence score is not a transferable probability of causation. A reference implementation is available from the authors on request.

\bibliographystyle{plainnat}
\bibliography{references}

\appendix

\section{Method Comparison Matrix}\label{app:methods}

Table~\ref{tab:method_full} provides the full comparison of all eleven MCES methods.

\begin{table}[H]
\centering
\caption{Complete method comparison: what each method captures, misses, and assumes. The final column lists the \emph{optional} tiered weight; the default weighting is uniform.}
\label{tab:method_full}
\small
\begin{tabularx}{\textwidth}{@{}P{2.4cm}P{2.8cm}P{2.8cm}P{2.8cm}c@{}}
\toprule
\textbf{Method} & \textbf{Captures} & \textbf{Misses} & \textbf{Key Assumption} & \textbf{Optional tiered $w_k$} \\
\midrule
Partial Corr. & Conditional linear assoc.\ after adjustment & Non-linear, direction & Linearity & 0.06 \\
Lasso & Sparse selection & Non-linear, temporal & Linearity, sparsity & 0.06 \\
Distance Corr. & Arbitrary-form dependence & Direction, confounders & Exchangeability (perm.\ test) & 0.06 \\
Mixed Effects & Unit baselines, panel & Non-linear, direction & Linearity & 0.07 \\
RF + SHAP & Non-linear, interactions & Direction, temporal & Predictive validity, exchangeability & 0.08 \\
PPS & Univariate OOS skill, asymmetry & Confounders, temporal & Representative folds, exchangeability & 0.05 \\
Granger & Temporal predictability & Non-linear & Stationarity & 0.14 \\
ITS & Intervention impact & Cross-sectional & Known intervention & 0.14 \\
Transfer Entropy & Non-linear direction & Small-sample sensitive & Sufficient time series & 0.13 \\
Bayesian Net & Adjacency, short paths & Hidden confounders, orientation & Causal sufficiency & 0.11 \\
Causal Forest & Heterogeneous effects & Continuous dose-response (this impl.) & Unconfoundedness, overlap, SUTVA & 0.10 \\
\bottomrule
\end{tabularx}
\end{table}

\section{Synthetic Data Generation}\label{app:synthetic}

The synthetic generator creates panel data $\mathbf{O} \in \R^{N \times T \times (M+K)}$ with embedded structure via one of five data-generating processes, matching the implementation:

\begin{enumerate}[nosep]
  \item \textbf{Drivers.} Numeric drivers are drawn per unit from truncated normals within catalog ranges, with AR(1) temporal dynamics ($\rho = 0.85$) for non-static categories.
  \item \textbf{True causal relationships (linear).} For each true edge, $y_{l} \mathrel{+}= s_j\, d_j\, z(x_j)$, where $z(\cdot)$ standardizes the driver, $s_j \in [0.15,0.6]$ is the strength, and $d_j \in \{+1,-1\}$ the sign.
  \item \textbf{Non-linear effects.} Before scaling, the standardized driver may be passed through one of four transforms: $z^2{-}1$ (quadratic), $\mathrm{sgn}(z)\sqrt{|z|}$ (sqrt), $\mathrm{sgn}(z)\log(1{+}|z|)$ (log), or $\mathbb{I}[z{>}0]{-}\tfrac12$ (threshold).
  \item \textbf{Time lags.} Effects may be applied at lag $\tau \in \{1,2,3\}$ within each unit.
  \item \textbf{Confounders.} Hidden variables injected into a subset of drivers \emph{and} outcomes, creating spurious associations.
  \item \textbf{Structural identities.} For the decomposition experiment, an outcome is generated as the exact product of two component drivers (Section~\ref{sec:exp_decomp}).
  \item \textbf{Noise.} $\varepsilon_{i,t} \sim \mathcal{N}(0, \sigma^2)$ with scenario-specific $\sigma$ (typically $0.2$--$0.35$).
\end{enumerate}

The ground-truth edge set is known by construction, enabling exact Precision@$K$, Recall@$K$, F1@$K$, Spearman $\rho$, and calibration.

\end{document}